\pdfoutput=1

\documentclass[journal]{IEEEtran}

\usepackage{graphicx}
\usepackage{amsmath}
\usepackage{amsfonts}
\usepackage{amssymb}
\usepackage{dsfont}
\usepackage{bm}
\usepackage{amsthm}
\usepackage{newlfont}
\usepackage{float}
\usepackage{hyperref}
\usepackage{algorithm}
\usepackage{algorithmic}
\usepackage{enumerate}
\usepackage{chngcntr}
\usepackage{mathtools}
\usepackage{breqn}
\usepackage{stmaryrd}
\usepackage{glossaries}
\usepackage{amsthm}
\usepackage{cite}
\usepackage{pgfplots}
\newlength\figureheight
\newlength\figurewidth

\usepackage[utf8x]{inputenc} % Pour les caractères accentués
\hypersetup{
    colorlinks=true, %colorise les liens
    breaklinks=true, %permet le retour à la ligne dans les liens trop longs
    urlcolor= black, %couleur des hyperliens
    linkcolor= blue, %couleur des liens internes
    citecolor=blue, %couleur des références
    pdftitle={report}, %informations apparaissant dans
    pdfauthor={Remi CHOU}, %les informations du document
    pdfsubject={Simulation} %sous Acrobat.
}

\begin{document}

\newtheorem{thm}{Theorem} [section]
\renewcommand{\thethm}{\arabic{section}.\arabic{thm}}
\newtheorem{lem}{Lemma} [section]
\renewcommand{\thelem}{\arabic{section}.\arabic{lem}}
\newtheorem{prop}{Proposition}[section]
\renewcommand{\theprop}{\arabic{section}.\arabic{prop}}
\newtheorem{cor}{Corollary} [section]
\renewcommand{\thecor}{\arabic{section}.\arabic{cor}}
\newtheorem{defn}{Definition} [section]
\renewcommand{\thedefn}{\arabic{section}.\arabic{defn}}
\newtheorem{rem}{Remark} [section]
\renewcommand{\therem}{\arabic{section}.\arabic{rem}}
\newtheorem{ex}{Example} [section]
\renewcommand{\theex}{\arabic{section}.\arabic{ex}}

\newenvironment{example}[1][Example]{\begin{trivlist}
\item[\hskip \labelsep {\bfseries #1}]}{\end{trivlist}}

\renewcommand{\qedsymbol}{ \begin{tiny}$\blacksquare$ \end{tiny} }

\renewcommand{\leq}{\leqslant}
\renewcommand{\geq}{\geqslant}

%-----------------------------------------%
% TABLE OF CONTENTS
%-----------------------------------------%

\title {Separation of Reliability and Secrecy\\ in Rate-Limited Secret-Key Generation}
\author{R\'{e}mi A. Chou,~\IEEEmembership{Student Member,~IEEE} and Matthieu R. Bloch,~\IEEEmembership{Member,~IEEE}
\thanks{This work was supported in part by the NSF under Award CCF 1320298 and the ANR grant ANR-13-BS03-0008.}
\thanks{R. A. Chou and M. R. Bloch are with the School~of~Electrical~and~Computer~Engineering,~Georgia~Institute~of~Technology, Atlanta,~GA~30332--0250 and GT-CNRS UMI 2958, 2 rue Marconi, 57070 Metz, France.} \thanks{E-mail : remi.chou@gatech.edu; matthieu.bloch@ece.gatech.edu. Preliminary versions of the results were presented at the 2012 IEEE International Symposium on Information Theory in \cite{Chou12}.}
%\thanks{Copyright (c) 2013 IEEE. Personal use of this material is permitted.  However, permission to use this material for any other purposes must be obtained from the IEEE by sending a request to pubs-permissions@ieee.org.}
}

\maketitle

\begin{abstract}
For a discrete or a continuous source model, we study the problem of secret-key generation with one round of rate-limited public communication between two legitimate users. Although we do not provide new bounds on the wiretap secret-key (WSK) capacity for the discrete source model, we use an alternative achievability scheme that may be useful for practical applications. As a side result, we conveniently extend known bounds to the case of a continuous source model. Specifically, we consider a sequential key-generation strategy, that implements a rate-limited reconciliation step to handle reliability, followed by a privacy amplification step performed with extractors to handle secrecy. We prove that such a sequential strategy achieves the best known bounds for the rate-limited WSK capacity (under the assumption of degraded sources in the case of two-way communication). However, we show that, unlike the case of rate-unlimited public communication, achieving the reconciliation capacity in a sequential strategy does not necessarily lead to achieving the best known bounds for the WSK capacity. Consequently, reliability and secrecy can be treated successively but not independently, thereby exhibiting a limitation of sequential strategies for rate-limited public communication. Nevertheless, we provide scenarios for which  reliability and secrecy can be treated successively and independently, such as the two-way rate-limited SK capacity, the one-way rate-limited WSK capacity for degraded binary symmetric sources, and the one-way rate-limited WSK capacity for Gaussian degraded sources. 
\end{abstract}
\begin{IEEEkeywords}
Secret-key capacity, secret-key generation, rate-limited communication, reconciliation, privacy amplification 
\end{IEEEkeywords}
\section{Introduction}

Information-theoretic secret-key generation protocols \cite{Maurer93,Ahlswede93} draw their strength from a security relying on information-theoretic metrics rather than on complexity theory, thereby avoiding the assumption of limited computational power for the eavesdropper. In such protocols, two legitimate users (Alice and Bob) and an eavesdropper (Eve) observe the realizations of correlated random variables (RVs), discrete or continuous. The legitimate users, who can exchange messages over a public channel, aim at extracting a common secret key from their observations. The rules by which the legitimate users compute the messages they exchange over the public channel and agree on a key define a key-generation strategy. The maximum number of secret-key bits per observed realization of the RVs is called the wiretap secret-key (WSK) capacity~\cite{Ahlswede93,Maurer93}. 

Closed-form expressions and bounds for the WSK capacity have been established for a large variety of models~\cite{Maurer93,Ahlswede93,Csiszar00,Csiszar04,Ye05,Csiszar08,Csiszar10,
Nitinawarat12,Csiszar13}. However, usual achievability proofs rely on a random binning argument and thus, do not always provide direct insight into the design of practical key-generation strategies. Moreover, such proofs handle reliability (the legitimate users must share the same key) and secrecy (the key must be unknown to the eavesdropper) jointly, which might limit the flexibility of the scheme. 

A more constructive scheme is a sequential key-generation strategy, which consists of two steps that handle reliability and secrecy successively instead of jointly. A reconciliation step~\cite{Brassard94} is first performed, during which Alice and Bob communicate over the public channel to agree on a common bit sequence, which might not be totally hidden from Eve. Then, a privacy amplification step~\cite{Bennett95,Maurer00} is performed, during which Alice and Bob apply a deterministic function to their shared sequence to generate their common secret key, this time completely unknown from Eve.
The main benefit of sequential key-generation strategies is to separate how one deals with reliability and secrecy,\footnote{We mean that the key-generation can be performed by the succession of two protocols, one, free from any secrecy constraint, dealing with reliability, and the other dealing with secrecy. A stronger result would be that optimizing both protocols independently, in a sense defined in Section~\ref{SecIndep}, leads to the best possible key-generation strategy. In Section~\ref{Sec_indep}, we prove that this stronger result holds in some scenarios.} and thus to provide a perhaps more practical key-generation design. Indeed, reconciliation can be efficiently implemented with LDPC codes~\cite{Bloch06b,Elkouss09}  and privacy amplification can be performed with extractors~\cite{Bennett95,Maurer00}. While sequential key-generation is studied in~\cite{Maurer00,Bloch11} for a public channel of unlimited capacity, we focus on the performance of sequential key-generation strategies in the case of rate-limited public communication.\footnote{Note that the achievability scheme of~\cite[Theorem 4.1]{Nitinawarat12}, which only holds for Gaussians sources and when there is no side information at the eavesdropper, is very close to the sequential approach that we study, even though their model is different in that it deals with a quantized source and unrestricted public communication.} Note that sequential strategies have also been used for secrecy purposes in \cite{Bellare12}, in which a practical capacity-achieving scheme involving invertible extractors is proposed for the symmetric wiretap channel.  % 

Besides sequential strategies, constructive secret-key capacity achieving schemes relying on polar codes have been recently proposed~\cite{Chou13b,Sutter13} for some of the models studied in this paper. A brief comparison between sequential schemes and polar code schemes can be found in \cite{Chou13b}.

  Although, we do not improve the rate-limited WSK capacity bounds for the discrete source model, we provide an achievability scheme that might be easier to translate into practical designs. Specifically, we show that sequential strategies, that are known to be optimal for rate-unlimited public communication, are also optimal for rate-limited communication. We however also qualify the robustness of sequential strategies to rate-limited public communication, as we show in this case that it may not be optimal to achieve the reconciliation capacity in a sequential strategy. That is, reliability and secrecy can be handled successively but not necessarily independently, thereby limiting the coding scheme flexibility. 
 The main contributions of this work are:
\begin{itemize}
\item  an alternative achievability scheme that separates reliability and secrecy by means of a reconciliation protocol and a privacy amplification step performed with extractors, which achieves
\begin{enumerate}[(i)]
\item the best known bound of the two-way one-round rate-limited WSK capacity for degraded sources in Theorem~\ref{theorem_seq2};
\item the  one-way rate-limited WSK capacity in Theorem~\ref{theorem_seq1} (this extends \cite{Chou12}, which only considers degraded sources) ;
\item the two-way one-round rate-limited SK capacity (no side information at the eavesdropper) in Theorem~\ref{theorem_seq3};
\end{enumerate}
As a side result, we extend the bounds for a discrete source model in~\cite{Csiszar00}, to the case of a continuous source model in Corollary~\ref{cor_extend} (the case of the one-way rate-limited WSK capacity is treated in \cite{Watanabe10a}, but only for degraded sources) ;
\item scenarios for which achieving the reconciliation capacity is optimal in a sequential key-generation strategy, as it is not necessarily the case in general when constraints are imposed on public communication. Such results are important to obtain a flexible coding scheme; Specifically, we treat the case of 
\begin{enumerate}[(i)]
\item the two-way rate-limited SK capacity in Section~\ref{Example_SK};
\item the one-way rate-limited WSK capacity for degraded binary symmetric sources in Section~\ref{Example_bin};
\item the one-way rate-limited WSK capacity for degraded Gaussian sources in Section~\ref{Example_gauss};
\end{enumerate}
As side results, we obtain a characterization of the rate-limited reconciliation capacity in Proposition~\ref{C_rec2}, which corresponds to the best trade-off between the length of the sequence shared by Alice and Bob after reconciliation and the quantity of information publicly exchanged; we also obtain a closed-form expression of the one-way WSK capacity for degraded binary symmetric sources with Proposition~\ref{discretegen}, as illustrated in  Example~\ref{Example_binary}.
\end{itemize}

Our proofs techniques mainly rely on the analysis of randomness extraction with extractors, Wyner-Ziv coding, and a fine analysis with robust typicality~\cite{Orlitsky01} to extend the discrete case to a continuous setting. The determination of the one-way WSK capacity for degraded binary symmetric sources relies on perhaps less standard techniques, as we use the Krein-Milman Theorem to simplify a convex optimization problem under convex constraints.

We note that our model and our analysis rely on restrictive assumptions.
\begin{itemize}
\item Our analysis only deals with asymptotic rates. We thus do not provide directly applicable results, but rather some insight for practical applications into the optimality or non-optimality of a sequential strategy in the case of rate-limited public communication. Specifically, we show that sequential strategies remain optimal in most cases for rate-limited communication, but unlike the case of rate-unlimited communication, achieving the reconciliation capacity in a sequential strategy is not necessarily optimal. 
\item We consider an independent and identically distributed (i.i.d.) source. While this assumption is not necessarily satisfied in practical settings, it remains realistic as shown in \cite{Pierrot13}, in which an i.i.d. source is induced in an indoor wireless environment.% for secret-key distillation purposes.
\item We assume the existence of an authenticated public channel. In practice, a solution to ensure this would be to have the legitimate users share a secret seed, which size can be chosen in the order of the logarithm of the length of the message~\cite{Wegman81,Maurer93}.
\end{itemize}
Note that for the finite-length regime in the case of rate-unlimited communication, a lower and an upper bound for achievable secret-key rates are provided in~\cite{Renner05} for an arbitrary source (an analysis of privacy amplification is also provided in~\cite{Watanabe12}), and an achievable secret-key rate is derived in \cite{Pierrot13} for i.i.d. sources. However, whether a sequential strategy is optimal for the finite-length regime remains an open question.

Note also that there exists works dealing with related models, that do not require such assumptions. For instance,~\cite{Dodis06} provides a non-asymptotic practical secret key-generation scheme for a non-memoryless source model, in which the legitimate users observe discrete components that are close with respect to certain metrics, and the eavesdropper has no observations of the source, with the assumption of one-way public communication over an unauthenticated channel with unlimited capacity. Privacy amplification is also treated for a non-asymptotic regime in~\cite{Dodis09} by means of malleable extractors~\cite{Dodis09,Cohen12}, when the legitimate users observe the same component of a non-memoryless source and the eavesdropper observes a correlated component of the source, with the assumption of two-way one-round communication over an unauthenticated channel with unlimited capacity. Such models are less general that the ones studied in this paper since the observations of the legitimate users are constrained to be equal or close to each other; however, they offer in return more practical solutions as they free themselves from the aforementioned assumptions.
The remainder of the paper is organized as follows. In Section \ref{Sec2}, we introduce the problem and provide some background on the topic. Specifically, we formally introduce the problem studied in Section \ref{SecMod}, and recall known bounds for the secret-key capacity in Section~\ref{SecKres}. In Section~\ref{SecStrat}, we describe the two steps of a sequential strategy and recall known bounds achieved by such a strategy. In Section~\ref{SecIndep}, we introduce the notion of independence between the two steps of a sequential strategy, when constraints are imposed on public communication. In Section \ref{Seqsec}, we prove that the sequential application of reconciliation and privacy amplification with extractors is an optimal key-generation strategy. In Section~\ref{Sec_indep}, we provide scenarios for which these two phases can be treated independently of each other. Specifically, we provide the case of the two-way SK capacity in Section~\ref{Example_SK}, the one-way WSK capacity for degraded binary symmetric sources in Section~\ref{Example_bin}, and the one-way WSK capacity for degraded Gaussian sources in Section~\ref{Example_gauss}. All proofs are gathered in the appendices to streamline presentation.

\section{Notation}
Consider $p,q \in \mathbb{R}$. We define the following  associative and commutative operation $p\star q \triangleq p(1-q)+(1-p)q$; observe that $[0,1]$ is closed with respect to $\star$. We define the integer interval $\llbracket p,q \rrbracket$, as the set of integers between $\lfloor p \rfloor$ and $\lceil q \rceil$. We define $[p]^+$ as $\max(0,p)$. We denote the Bernoulli distribution with parameter $p \in [0,1]$ by $\mathcal{B}(p)$, and for any $p\in[0,1]$, we define $\bar{p} \triangleq 1-p$. Finally, we note ${H}_b(\cdot)$ the binary entropy function.%, and $(\mathcal{B}_c(K), || .||_{\infty})$ the set of $K$-bounded continuous function, where $K \in \mathbb{R}$.
\section{Problem Statement and Background} \label{Sec2}
\subsection{Source {M}odel for {S}ecret {K}ey-{G}eneration} \label{SecMod}
As illustrated in Figure \ref{figsource}, a source model for secret-key generation represents a situation in which two legitimate users, Alice and Bob, and one eavesdropper, Eve, observe the realizations of a memoryless source (MS) $(\mathcal{X}\mathcal{Y}\mathcal{Z},p_{XYZ})$, that can be either discrete (DMS) or continuous (CMS). The three components $X$, $Y$ and $Z$, are observed by Alice, Bob, and Eve, respectively. The MS is assumed to be outside the control of all parties, but its statistics are known. Alice and Bob's objective is to process their observations and agree on a key $K$, about which Eve should have no information. We assume a two-way one-round communication between Alice and Bob, that is, we suppose that Alice first sends a message to Bob, and that in return Bob sends a message to Alice.\footnote{One could also suppose that Bob is the one who sends messages, in which case one only needs to exchange the role of $X$ and $Y$ in the following.} We also assume that the messages  are exchanged over an authenticated noiseless public channel with limited rate; in others words, Eve has total access to Alice and Bob's messages, but cannot tamper with the messages over the channel. 
\begin{figure}%[H]
\centering
  \includegraphics[width=8.2cm]{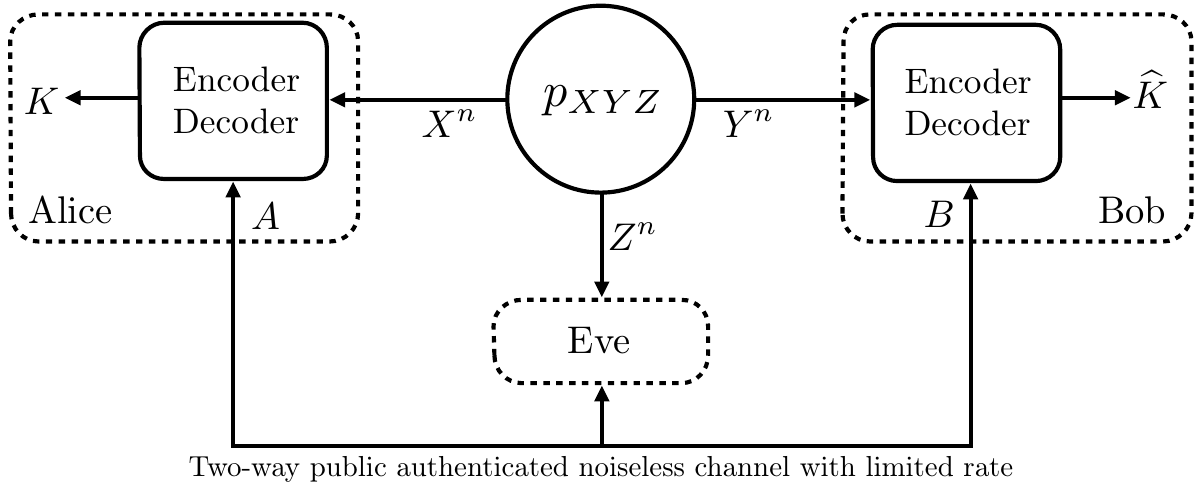}
  \caption{ Source model for secret-key generation}
  \label{figsource}
\end{figure}
We now formally define a key-generation strategy.
\begin{defn}
A $\left(2^{nR},n,R_1,R_2\right)$ key-generation strategy $\mathcal{S}_n$ for a source model with MS $(\mathcal{X}\mathcal{Y}\mathcal{Z},p_{XYZ})$ consists of
\begin{itemize}
\item a key alphabet $\mathcal{K} = \left\llbracket 1, 2^{nR} \right\rrbracket$;
\item two alphabets $\mathcal{A}$, $\mathcal{B}$ respectively used by Alice and Bob to communicate over the public channel;
\item two encoding functions $f_0 : \mathcal{X}^n \rightarrow \mathcal{A}$, $g_0 : \mathcal{Y}^n \times \mathcal{A} \rightarrow \mathcal{B}$;
\item two functions $\kappa_a: \mathcal{X}^n \times \mathcal{B} \rightarrow \mathcal{K}$, $\kappa_b: \mathcal{Y}^n \times \mathcal{A} \rightarrow \mathcal{K}$;
\end{itemize}
and operates as follows.
\begin{itemize}
\item Alice observes $X^n$ while Bob observes $Y^n$;
\item Alice transmits $A = f_0 (X^n)$ subject to ${H}(A) \leq n R_1$;
\item Bob transmits $B = g_0 (Y^n,A)$ subject to ${H}(B) \leq n R_2$;
\item Alice computes $K= \kappa_a(X^n,B)$ while Bob computes $\hat{K} = \kappa_b(Y^n,A)$.
\end{itemize}
\end{defn}
The performance of a $\left(2^{nR},n,R_1,R_2\right)$ key-generation strategy $\mathcal{S}_n$ is measured in terms of the average probability of error between the key $K$ generated by Alice and the key $\hat{K}$ generated by Bob
$$\textbf{P}_e(\mathcal{S}_n) \triangleq \mathbb{P} [ K \neq \hat{K} | \mathcal{S}_n ],
$$ in terms of the information leakage to the eavesdropper $$\textbf{L}(\mathcal{S}_n) \triangleq {I} (K;Z^nAB|\mathcal{S}_n),$$ and in terms of the uniformity of the key $$\textbf{U}(\mathcal{S}_n) \triangleq \log \left\lceil 2^{nR} \right\rceil - {H}(K| \mathcal{S}_n).
$$
\begin{defn} \label{defsec}
A WSK rate $R$ is achievable for a source model if there exists a sequence of $\left(2^{nR},n,R_1,R_2\right)$ key-generation strategies $\left\{ \mathcal{S}_n \right\}_{n \geq 1}$ such that
\begin{align*}
& \displaystyle\lim_{n \to \infty } \textbf{\textup{P}}_e(\mathcal{S}_n) =0  && \text{(reliability),}\\
& \displaystyle\lim_{n \to \infty } \textbf{\textup{L}}(\mathcal{S}_n) =0  && \text{(strong secrecy),}\\
& \displaystyle\lim_{n \to \infty } \textbf{\textup{U}}(\mathcal{S}_n)=0  & &\text{(strong uniformity).}
\end{align*}
Moreover, the WSK capacity of a source model with MS $(\mathcal{X}\mathcal{Y}\mathcal{Z},p_{XYZ})$ is the supremum of achievable WSK rates, and is denoted by $C_{\textup{WSK}}$.
In the following, we also consider situations in which the eavesdropper has access to the  public messages exchanged by Alice and Bob, but has no side information $Z^n$. In such cases, the WSK capacity is simply called the secret-key (SK) capacity and is denoted by $C_{\textup{SK}}$.
\end{defn}
\begin{rem}
As shown in \cite{Bellare12}, the security criterion $\textbf{\textup{L}}(\mathcal{S}_n)$ in Definition \ref{defsec} is equivalent to a definition proposed in~\cite{Bellare12} and inspired by semantic security for computationally bounded eavesdropper~\cite{Goldwasser84}. In our case, the definition in \cite{Bellare12} translates to 
$$
 \max_{\mathcal{A},f} \left( \mathbb{P} \left[ \mathcal{A}(Z^n,A,B) = f(K) \right] - \frac{1}{|f(K)|} \right),
$$
where the maximization is over any computationally unbounded adversary $\mathcal{A}$ and any function $f$. It means that for any $\mathcal{A}$ and $f$, the probability that an adversary $\mathcal{A}$ recovers $f(K)$ from $(Z^n,A,B)$ should not be better than the probability of guessing $f(K)$, only knowing its length $|f(K)|$. 
\end{rem}
\subsection{Known Bounds for $C_{\textup{WSK}}$ and $C_{\textup{SK}}$} \label{SecKres}
For convenience, we recall known results regarding the model described in Section~\ref{SecMod}. Note that these results  only hold for DMS.\\

\begin{thm}[\! {\cite[Theorems 2.5, 2.6]{Csiszar00}}]  \label{C_WSK} 
Let $(\mathcal{X}\mathcal{Y}\mathcal{Z},p_{XYZ})$ be a DMS. 
\counterwithin{enumi}{prop}\begin{enumerate}[(a)]
\item For $R_1,R_2 \in \mathbb{R}^+$, the two-way one-round WSK capacity satisfies \label{C_WSK2}
\begin{align*}
& C_{\textup{WSK}}(R_1,R_2) \geq  R_{\textup{WSK}}(R_1,R_2),
\end{align*}
where
\begin{multline*}
R_{\textup{WSK}}(R_1,R_2) \triangleq \displaystyle\max_{U,V} \left( [{I}(Y;U) -{I}(Z;U)]^+ \right. \\ \left. + [ {I}(X;V|U) - {I}(Z;V|U)]^+ \right) %\text{ subject to }
\end{multline*}   
\vspace*{-1em}
subject to
\vspace*{-1.1em}
\begin{align*} 
&R_1 \geq {I}(X;U) - I(Y;U),  \\ %\label{rate1}
&R_2 \geq  {I}(Y;V|U) - I(X;V|U), \\ \nonumber %\label{rate2}
&U \text{---} X \text{---} YZ  \text{, } V \text{---} YU \text{---} XZ, \\  \nonumber 
& |\mathcal{U}| \leq |\mathcal{X}|+2, |\mathcal{V}| \leq |\mathcal{Y}|. %\label{constrs}
\end{align*}
\item For $R_1 \in \mathbb{R}^+$, the one-way WSK capacity is \label{C_WSK1}
\begin{align*}
& C_{\textup{WSK}}(R_1,0) = \displaystyle\max_{U,V} \left( {I}(Y;V|U) - {I}(Z;V|U) \right) \text{ }
\end{align*}
\vspace*{-1em}
subject to
\vspace*{-1.1em}
\begin{align}
&R_1 \geq {I}(X;V) - I(Y;V), \nonumber \\
&U \text{---} V \text{---} X \text{---} YZ, \nonumber \\%\label{markov3} \\
&|\mathcal{U}|,|\mathcal{V}| \leq |\mathcal{X}|+2.\nonumber %\label{constrplusoneb}
\end{align}
\end{enumerate}
\end{thm}
\begin{cor}[\! {\cite[Theorems 2.2, 2.3, 2.4]{Csiszar00}}]  \label{C_SK}
Let $(\mathcal{X}\mathcal{Y},p_{XY})$ be a DMS. 
\counterwithin{enumi}{prop}\begin{enumerate}[(a)] \label{C_SK2}
\item For $R_1,R_2 \in \mathbb{R}^+$, the two-way one-round SK capacity is
\begin{align*}
& C_{\textup{SK}}(R_1,R_2) = \displaystyle\max_{U,V} \left({I}(Y;U)  + {I}(X;V|U) \right) %\text{ subject to}
\end{align*}
\vspace*{-1em}
subject to
\vspace*{-1.1em}
\begin{align*}
&R_1 \geq {I}(X;U) - I(Y;U), \nonumber \\
&R_2 \geq  {I}(Y;V|U) - I(X;V|U), \nonumber\\ 
&U \text{---} X \text{---} Y\text{, } V \text{---} YU \text{---} X,  \\
&|\mathcal{U}| \leq |\mathcal{X}|+2, |\mathcal{V}| \leq |\mathcal{Y}|.\nonumber %\label{constrplusoneb}%&|\mathcal{U}| \leq |\mathcal{X}|+2, |\mathcal{V}| \leq |\mathcal{Y}| \label{constrplusoneb}
\end{align*}
%rate constraints (\ref{rate1}),  (\ref{rate2}), and range constraints (\ref{constrplusoneb}).
\item For $R_1 \in \mathbb{R}^+$, the one-way SK capacity is \label{C_SK1}
\begin{align}
& C_{\textup{SK}}(R_1,0) = \displaystyle\max_{U} {I}(Y;U) \nonumber 
\end{align}
\vspace*{-1em}
 subject to
\vspace*{-1.1em}
\begin{align*}
&R_1 \geq {I}(X;U) - I(Y;U), \\%\label{rate1} \\
&U \text{---} X \text{---} Y, \nonumber \\
&|\mathcal{U}| \leq |\mathcal{X}|+1.%\label{constrplusone}
\end{align*}
%rate constraint (\ref{rate1}), and Markov condition (\ref{markov1}).%, and range constraint (\ref{constrplusoneb}).
\end{enumerate}
\end{cor}
\subsection{Sequential Strategy} \label{SecStrat}
In the following, we use the term sequential key-generation strategy, for a key-generation strategy consisting of the succession of a reconciliation protocol and a privacy amplification with extractors. %
\subsubsection{Reconciliation}\label{reconciliationsec}
During the reconciliation phase, Alice and Bob send messages to each other over an authenticated public channel with limited rate. Alice and Bob then process their observations to agree on a common bit sequence $S$. At this stage the sequence is not subject to any secrecy constraint. %
Formally, a two-way one-round rate-limited reconciliation protocol is defined as follows.
\begin{defn} \label{reconciliation}
Let $R_1,R_2 \in \mathbb{R}^+$. A rate-limited reconciliation protocol $\mathcal{R}_n(R_1,R_2)$, noted $\mathcal{R}_n$ for convenience, for a source model with MS $(\mathcal{X}\mathcal{Y},p_{XY})$ consists of
\begin{itemize}
\item an alphabet $\mathcal{S} = \llbracket 1, M\rrbracket$;
\item two alphabets $\mathcal{A}$, $\mathcal{B}$ respectively used by Alice and Bob to communicate over the public channel;
\item two encoding functions $f : \mathcal{X}^n \rightarrow \mathcal{A}$, $g : \mathcal{Y}^n \times \mathcal{A} \rightarrow \mathcal{B}$;
\item two functions $\eta_a: \mathcal{X}^n \times \mathcal{B} \rightarrow \mathcal{S}$, $\eta_b: \mathcal{Y}^n \times \mathcal{A} \rightarrow \mathcal{S}$;
\end{itemize}
and operates as follows
\begin{itemize}
\item Alice observes $X^n$ while Bob observes $Y^n$;
\item Alice transmits $A = f (X^n)$ subject to ${H}(A) \leq n R_1$;
\item Bob transmits $B = g (Y^n,A)$ subject to ${H}(B) \leq n R_2$;
\item Alice computes $S = \eta_a(X^n,B)$ while bob computes $\hat{S} = \eta_b(Y^n,A)$.
\end{itemize}
\end{defn}
The reliability performance of a reconciliation protocol is measured in terms of the average probability of error $$\textbf{P}_e(\mathcal{R}_n) \triangleq \mathbb{P} [ S \neq \hat{S} | \mathcal{R}_n ].
$$
In addition, since the reconciliation protocol, which generates the common sequence $S$, is followed by the privacy amplification step to generate a secret-key, it is desirable to leak as little information as possible over the public channel. As in~\cite{Bloch11} we define the reconciliation rate of a reconciliation protocol as 
$$\textbf{R}(\mathcal{R}_n) \triangleq \frac{1}{n} \left[ {H}(S|\mathcal{R}_n) - {H}(AB|\mathcal{R}_n)\right].
$$%
\begin{defn}
For a given $(R_1,R_2)$, a reconciliation rate $R$ is achievable, if there exists a sequence of  rate-limited reconciliation protocols $\left\{ \mathcal{R}_n \right\}_{n \geq 1}$ such that
\begin{equation*}
\displaystyle\lim_{n \rightarrow \infty} \textup{\textbf{P}}_e(\mathcal{R}_n) = 0  \text{ and }\displaystyle \varliminf_{n\rightarrow \infty} \textup{\textbf{R}}(\mathcal{R}_n)  \geq R.
\end{equation*}
Moreover, the two-way one-round rate-limited reconciliation capacity $C_{\textup{rec}}(R_1,R_2)$ of a MS $(\mathcal{X}\mathcal{Y},p_{XY})$ is the supremum of achievable reconciliation rates.
\end{defn}
Intuitively, the reconciliation capacity characterizes the best trade-off between the length of the sequence shared by Alice and Bob after reconciliation and the quantity of information publicly exchanged.
\subsubsection{Privacy amplification} \label{privateamp}
During the privacy amplification phase, Alice and Bob generate their secret key by applying a deterministic function, on which they publicly agreed ahead of time, to their common sequence $S$ obtained after reconciliation. This phase is performed with extractors~\cite{Vadhan98}, which are functions that take as input a sequence of $n$ arbitrarily distributed bits and output a sequence of $k$ nearly uniformly distributed bits, using another input of $d$ truly uniformly distributed bits. The following theorem provides a lower bound on the size of the key, on which the legitimate users agree.
\begin{thm} [\!\cite{Maurer00}, {\cite[Theorem 4.6]{Bloch11}}] \label{thMaurer}
Let $S \in \left\{ 0,1 \right\}^n$ be the RV that represents the common sequence shared by Alice and Bob, and let $E$ be the RV that represents the total knowledge about $S$ available to Eve. Let $e$ be a particular realization of $E$. 
If Alice and Bob know that $${H}_{\infty}(S|E=e) \geq \gamma n, \text{ for some $\gamma \in ]0,1[$,}$$ then there exists an extractor $g: \left\{0,1\right\}^n \times \left\{0,1\right\}^d \to \left\{0,1\right\}^k$ with 
$d \leq n \delta(n) \text{ and } k \geq n(\gamma - \delta(n))$, where $\delta(n)$ satisfies $\lim_{n \rightarrow + \infty} \delta(n) = 0$.\\
Moreover, if $U_d$ is a RV uniformly distributed on $\left\{0,1\right\}^d$ and Alice and Bob choose $K=g(S,U_d)$ as their secret key, then
\begin{equation*}
{H}(K|U_d, E=e) \geq k - \delta^*(n), 
\end{equation*}
with  $\delta^*(n)=2^{-\sqrt{n}/\log n} \left( k+\sqrt{n}/\log n\right)$.
\end{thm}
Note that, the size $d$ of the uniformly distributed input sequence is negligible, compared to $n$, so that the effect on the rate of public communication is negligible. Moreover, extractors that extract almost the entire min-entropy of the input~$S$ and require comparatively negligible amount of uniform randomness can be efficiently constructed~\cite{Vadhan98}.
\subsubsection{Known Results Concerning Sequential Strategy}
For a DMS, in the absence of rate constraint between Alice and Bob, i.e. $R_1,R_2 = +\infty$, \cite{Maurer00}, \cite[Theorem 4.7]{Bloch11} state that one can handle reliability and secrecy successively to achieve the WSK capacity $C_{\textup{WSK}}(+\infty,+\infty)$, by means of a reconciliation step, and a privacy amplification step. This result is extended to the case of one-way rate-limited communication, i.e. $R_1 \in \mathbb{R}_+$ and $R_2 =0$, for a degraded DMS in \cite{Chou12}.
\subsection{Independence Between Reconciliation and Privacy Amplification} \label{SecIndep}
In this section, we define a notion of independence between reconciliation and privacy amplification, when constraints hold on the public communication rate. As explained earlier, we would like to ensure that reliability and secrecy can be handled not only successively but also independently, to obtain a flexible coding scheme.
We first recall that in the case of one-way rate-limited communication, the reconciliation capacity is given by the following proposition.
\begin{prop} [\!\!\cite{Chou12}] \label{prop_onerec}
Let $(\mathcal{X}\mathcal{Y},p_{XY})$ be a DMS. Let $R_1 \in \mathbb{R}_+$. The reconciliation capacity $C_{\textup{rec}}(R_1,0)$ is given by  
\begin{align*}
C_{\textup{rec}}(R_1,0) = C_{\textup{SK}}(R_1,0).
\end{align*}
\end{prop}
As shown in Example~\ref{counterexample}, unlike the case of rate-unlimited communication, in the case of rate-limited communication, it is not necessarily optimal to first achieve the reconciliation capacity in Proposition~\ref{prop_onerec} and then to perform privacy amplification. In other words, if a sequential strategy is known to achieve the secret-key capacity, it does not tell us at which rate we should perform the reconciliation step. In the following, we say that reconciliation and privacy amplification are independent if achieving the reconciliation capacity in a sequential strategy leads to achieving the secret-key capacity.\\

\begin{ex} \label{counterexample}
Consider the scenario presented in Figure \ref{figcountex}, in which $|\mathcal{X}| = |\mathcal{Y}|=|\mathcal{Z}|=2 $, $X \text{---} Y \text{---} Z$ forms a Markov chain, and $X \sim \mathcal{B}(p)$. We assume a one-way rate-limited public communication, i.e $R_1 \in \mathbb{R}$ and $R_2 = 0$. We set the parameters as follows. $R_1 = H(X|Y)/3$, $p=0.23$, $\beta_1 = 0.01 $, $\beta_2 = 0.03 $, $\gamma_1 =0.03 $ and $\gamma_2 = 0.01$.

 We know by \cite{Chou12} that a sequential strategy achieves the WSK capacity $C_{\textup{WSK}}(R_1,0)$. Moreover, we can show that 
 \begin{align}
& C_{\textup{WSK}}(R_1,0) = \max_{\alpha_1,\alpha_2} (f-g)(\alpha_1,\alpha_2),\nonumber \\ & \phantom{mmm}  \text{ subject to }  (h-f)(\alpha_1,\alpha_2) = R_1 \label{eqmax1},\\
& C_{\textup{rec}}(R_1,0) = \max_{\alpha_1,\alpha_2} f(\alpha_1,\alpha_2), \nonumber\\ 
& \phantom{mmm}  \text{ subject to } (h-f)(\alpha_1,\alpha_2) = R_1 \label{eqmax2},
 \end{align}
 \vspace*{-0.7em}
where
 \begin{align*}
  f(\alpha_1,\alpha_2) & \triangleq {H}_b( p_y) - p_uH_b(a) - \bar{p}_uH_b(b), \\
  g(\alpha_1,\alpha_2)& \triangleq H_b(p_z) -p_uH_b(c) - \bar{p}_uH_b(d),\\
  h(\alpha_1,\alpha_2)& \triangleq H_b(p) -p_u H_b(\alpha_1) - \bar{p}_u H_b(\alpha_2),
 \end{align*}
with
$p_u = (\bar{\alpha}_2 - p)/(\bar{\alpha}_2-\alpha_1)$, 
$p_y = \bar{p} \bar{\beta}_1+p\beta_2$,
$p_z = p_y \bar{\gamma}_1 + \bar{p}_y \gamma_2$, 
$a = \alpha_1 \beta_2+ \bar{\alpha}_1\bar{\beta}_1$, 
$b = \alpha_2 \bar{\beta}_1 + \bar{\alpha}_2\beta_2$, 
$c = \bar{\gamma}_1 a + \gamma_2 \bar{a}$,
%$c = \bar{\alpha}_1\bar{\beta}_1\bar{\gamma}_1 + \bar{\alpha}_1\beta_1 \gamma_2 + \alpha_1 \beta_2\bar{\gamma}_1 + \alpha_1 \bar{\beta}_2 \gamma_2$, 
$d = \bar{\gamma}_1 b + \gamma_2 \bar{b}$.
%$d = \bar{\alpha}_2 \bar{\beta}_2\gamma_2 + \bar{\alpha}_2 \beta_2 \bar{\gamma}_1  + \alpha_2\bar{\beta}_1\bar{\gamma}_1 + \alpha_2 \beta_1 \gamma_2$.
%

Numerically, $$C_{\textup{WSK}}(R_1,0) > 0.050 > 0.045 >  (f - g) (\alpha_1^*,\alpha_2^*),$$ where $(\alpha_1^*,\alpha_2^*)$ achieves $C_{\textup{rec}}(R_1,0)$.
% \triangleq \arg \! \max_{\alpha_1,\alpha_2} f(\alpha_1,\alpha_2)$, subject to $(h-f)(\alpha_1,\alpha_2) = R_1.$ 
 Hence, for this example, achieving the reconciliation capacity in a sequential key-generation is not optimal and incurs a rate loss above $10 \%$.
\end{ex}
\begin{rem}
Deriving (\ref{eqmax1}) and (\ref{eqmax2}) is not straightforward. We used Proposition~\ref{Sufprop} given in the following sections, which shows that equality holds in the public communication rate constraint (\ref{rate1}) and that $|\mathcal{U}| \leq |\mathcal{X}|$.% in Theorem~\ref{C_WSK} and Corollary~\ref{C_SK}.
\end{rem}
\begin{figure}
\begin{center}
  \includegraphics[width=8.2cm]{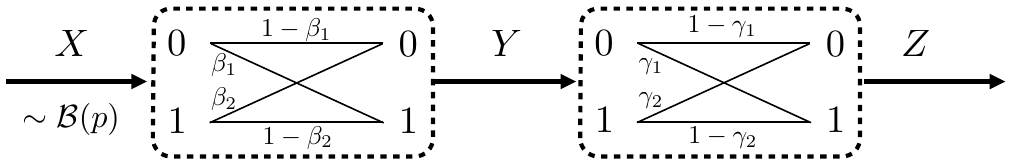}
  \caption{Example of a binary DMS studied in Example \ref{counterexample}}
\label{figcountex}
\end{center}
\end{figure}
In Section~\ref{Seqsec}, for $R_1,R_2 \in \mathbb{R}_+$, we study the achievability of $R_{\textup{WSK}}(R_1,R_2)$, $C_{\textup{WSK}}(R_1,0)$, given in Theorem \ref{C_WSK} and $C_{\textup{SK}}(R_1,R_2)$ given in Corollary \ref{C_SK}, with a sequential key-generation strategy. Moreover, in Section~\ref{Sec_indep}, we identify scenarios for which reconciliation and privacy are independent in the sense defined in this section.
\section{Sequential Strategies Achieve the Best Know Bounds of $C_{\textup{WSK}}$ and $C_{\textup{SK}}$} \label{Seqsec}
In this section, we provide one of our main result. That is, the successive combination of reconciliation and privacy amplification, achieves the best known rates of the secret-key capacity  (under the assumption of degraded sources in the case of two-way communication), when constraints are imposed on the public communication. As a side result, we extend known bounds of $C_{\textup{WSK}}$ and $C_{\textup{SK}}$ for DMS to the case of CMS.
\begin{thm} \label{theorem_seq2}
Let $(\mathcal{X}\mathcal{Y}\mathcal{Z},p_{XYZ})$ be a MS such that $X \text{---} Y \text{---} Z$. For $R_1, R_2 \in \mathbb{R^+}$, all WSK rates $R$ that satisfy
\begin{equation*}
R < R_{\textup{WSK}}(R_1,R_2) 
\end{equation*}
are achievable with sequential key-generation strategies.
\end{thm}
\begin{IEEEproof}
See Appendix \ref{AppendixThcont}.
\end{IEEEproof}
\begin{rem}
Note that we assume $X \text{---} Y \text{---} Z$. For two-way communication, the necessity of this hypothesis might be an inherent weakness of a scheme that consists of a successive design of reconciliation and privacy amplification, rather than a joint design as in \cite{Csiszar00} (see the proof in Appendix~\ref{AppendixThcont} for more details). Observe, however, that for a one-way public communication, in Theorem \ref{theorem_seq1}, this assumption is not required.
\end{rem}
\begin{thm} \label{theorem_seq1}
Let $(\mathcal{X}\mathcal{Y}\mathcal{Z},p_{XYZ})$ be a MS. For $R_1 \in \mathbb{R^+}$, all WSK rates $R$ that satisfy
\begin{equation*}
R < C_{\textup{WSK}}(R_1,0) 
\end{equation*}
are achievable with sequential key-generation strategies.
\end{thm}
\begin{IEEEproof}
See Appendix \ref{AppendixTh1}.
\end{IEEEproof}
\begin{thm} \label{theorem_seq3}
Let $(\mathcal{X}\mathcal{Y},p_{XY})$ be a MS.
For $R_1,R_2 \in \mathbb{R^+}$, all SK rates $R$ that satisfy
\begin{equation*}
R < C_{\textup{SK}}(R_1,R_2) 
\end{equation*}
are achievable with sequential key-generation strategies.
\end{thm}
We omit the proof of Theorem~\ref{theorem_seq3}, which is similar to the one of Theorem~\ref{theorem_seq2} without the RV $Z$.\\
Note that putting constraints on the public communication leads to auxiliary random variables in the expression of the secret-key capacity and the reconciliation capacity, as seen in Section~\ref{Sec2}. Hence, as demonstrated in Example \ref{counterexample}, auxiliary random variables that achieve the reconciliation capacity, may not achieve the secret-key capacity. % sequential key-distillation 
In other words, reliability and secrecy can be handled successively, but cannot necessarily be treated independently, as defined in Section \ref{SecIndep}.
Nevertheless, in the next section, we identify scenarios for which reconciliation and privacy amplification can be treated independently.

As a side result, we have extended known bounds for the secret-key capacity for DMS to the case of CMS. We summarize this result in the following corollary, which is directly deduced from Theorems~\ref{theorem_seq2}, \ref{theorem_seq1}, and \ref{theorem_seq3}.
\begin{cor} \label{cor_extend}
Let $(\mathcal{X}\mathcal{Y}\mathcal{Z},p_{XYZ})$ be a MS. 
\counterwithin{enumi}{prop}\begin{enumerate}[(a)]
\item Assume that $X \text{---} Y \text{---} Z$. For $R_1,R_2 \in \mathbb{R}^+$, the two-way WSK achievable bound $R_{\textup{WSK}}(R_1,R_2)$ given in Theorem~\ref{C_WSK}.\ref{C_WSK2}, remains valid for a CMS.
\item For $R_1 \in \mathbb{R}^+$, the expression of the one-way WSK capacity $C_{\textup{WSK}}(R_1,0)$ given in Theorem \ref{C_WSK}.\ref{C_WSK1}, remains valid for a~CMS.
\item For $R_1,R_2 \in \mathbb{R}^+$, the two-way SK capacity $C_{\textup{SK}}(R_1,R_2)$ given in Corollary~\ref{C_SK}, remains valid for a CMS.
\end{enumerate}
\end{cor}
\section{Scenarios for Which Independence Holds Between Reliability and Secrecy} \label{Sec_indep}
As seen in the Example~\ref{counterexample}, achieving the reconciliation capacity might not lead to achieving the secret-key capacity. In this section, we identify special cases for which independence holds between reconciliation and privacy amplification.  Specifically, we prove that independence holds for the two-way one-round SK capacity, the one-way WSK capacity in the case of binary symmetric degraded sources, and the one-way WSK capacity in the case of Gaussian degraded sources. 
As a side result, we obtain an expression for the two-way rate-limited reconciliation capacity and a closed-form expression for the secret-key capacity $C_{\textup{WSK}}(R_1,0)$ in the case of degraded binary symmetric sources.
\subsection{Two-Way Rate-Limited SK capacity} \label{Example_SK}
In this section, we consider the two-way rate-limited SK capacity. That is, the eavesdropper has no correlated observation of the source.

We first show that the two-way rate-limited SK capacity is equal to the two-way rate-limited reconciliation capacity in the following proposition.
\begin{prop} \label{C_rec2}
Let $(\mathcal{X}\mathcal{Y},p_{XY})$ be a MS. 
 For $R_1,R_2\in \mathbb{R}^+$, the rate-limited reconciliation capacity $C_{\textup{rec}}(R_1,R_2)$ is  
\begin{align*}
C_{\textup{rec}}(R_1,R_2) = C_{\textup{SK}}(R_1,R_2).
\end{align*}
\end{prop}
\begin{IEEEproof}
See Appendix \ref{AppendixC_rec}.
\end{IEEEproof}
Hence, by Proposition \ref{C_rec2}, the auxiliary random variables that achieve the reconciliation capacity, also achieve the secret-key capacity; combined with Theorem~\ref{theorem_seq3}, we obtain the following corollary.
\begin{cor}
Let $(\mathcal{X}\mathcal{Y},p_{XY})$ be a MS and $R_1,R_2\in \mathbb{R}^+$. The two-way rate-limited SK capacity $C_{\textup{SK}}(R_1,R_2)$ is achievable by a sequential strategy, moreover, reconciliation and privacy amplification steps can be handled independently, as defined in Section~\ref{SecIndep}.
\end{cor}
\subsection{One-Way Rate-Limited WSK capacity for Degraded Binary Symmetric Sources} \label{Example_bin}
In this section, we assume a degraded DMS. We first refine Proposition~\ref{C_rec2} and Theorem~\ref{C_WSK}.\ref{C_WSK1}  in the following proposition.
\begin{prop} \label{Sufprop}
Let $(\mathcal{X}\mathcal{Y}\mathcal{Z},p_{XYZ})$ be a DMS such that $ X \text{---} Y\text{---} Z$. Assume $R_1 \in \mathbb{R}^+$ and $R_2=0$. We tighten the rate constraint in~(\ref{rate1}), (\ref{d_30}) and the range constraint in~(\ref{dddd}), (\ref{d_31}) as follows.
 \begin{enumerate} [(a)]
 \item \label{C_rec3} The one-way rate-limited reconciliation capacity is  
 \begin{align*}
& C_{\textup{rec}}(R_1,0)  = \displaystyle\max_{U} {I}(Y;U)  %\text{ subject to}
\end{align*}
\vspace*{-1em}
\text{ subject to }
\vspace*{-1em}
\begin{align}
&R_1 = {I}(X;U|Y),   \label{rate1}\\
& U \text{---} X \text{---} Y, \nonumber \\ 
&|\mathcal{U}| \leq |\mathcal{X}|. \label{dddd}
\end{align}
%\end{prop}
%
\item  \label{C_s_eq} The one-way rate-limited secret-key capacity is 
\begin{align*}
& C_{\textup{WSK}}(R_1,0) = \displaystyle\max_{U} \left( {I}(Y;U) - {I}(Z;U)\right) \end{align*}
\vspace*{-1.11em}
\text{ subject to }
\vspace*{-1em}
\begin{align} 
& R_1 = {I}(X;U|Y), \label{d_30} \\ \nonumber
&  U \text{---} X \text{---} Y\text{---} Z, \label{d_31}\\ 
&|\mathcal{U}| \leq |\mathcal{X}| .
\end{align}
\end{enumerate}
\end{prop}
\begin{IEEEproof}
See Appendix \ref{AppendixC_rec1}.
\end{IEEEproof}
\begin{rem}
The expression of the WSK capacity in Proposition \ref{Sufprop}.\ref{C_s_eq} is obtained from Theorem~\ref{C_WSK}.\ref{C_WSK1} and is due to Watanabe \cite{Watanabe10a}. We refine this result by proving that equality holds in the rate constraint and by improving the range constraint of $\mathcal{U}$; The argument used to show the equality in the rate constraint of  Propositions~\ref{Sufprop}.\ref{C_rec3} and \ref{Sufprop}.\ref{C_s_eq}, is one that applies to various convex maximization problems: the maximum principle (see Appendix~\ref{AppendixC_rec1}). This refinement is critical for the analysis of binary sources, especially to solve the optimization problem for the WSK capacity in Proposition~\ref{discretegen}, and thus to determine the WSK capacity for degraded binary symmetric sources in Example~\ref{Example_binary}.
\end{rem}
\begin{rem} \label{treshold_bin}
As soon as $R_1$ is at least ${H}(X|Y)$, $C_{\textup{rec}}(R_1,0)$ (resp.  $C_{\textup{WSK}}(R_1,0)$) attains the same maximum ${I}(X;Y)$ (resp. ${I}(X;Y) - {I}(X;Z)$) as in the case $R_1=+\infty$.
\end{rem}
%
%
%\subsection{Binary source}  
%
The solution of the maximization problem in Proposition~\ref{Sufprop}.\ref{C_s_eq} can be obtained explicitly, when the source has symmetry properties.

\begin{prop}\label{discretegen}
Let $(\mathcal{X}\mathcal{Y}\mathcal{Z},p_{XYZ})$ be a DMS such that $ X \text{---} Y\text{---} Z$. Assume that $|\mathcal{X}|=2$ and let $R_1 \in \mathbb{R}_+^*$.

 If the channels $p_{Y|X}$ and $p_{Z|X}$ are symmetric~\cite{Gallager68}, then the auxiliary RV $U$ achieving $C_{\textup{WSK}}(R_1,0)$ in Proposition \ref{Sufprop}.\ref{C_s_eq}, is such that the test-channel $p_{U|X}$ is a BSC with parameter $\beta_0$, with $\beta_0$, any of the two symmetric solutions of
\begin{equation*}
R_1 = {I}(U;X) - {I}(U;Y).
\end{equation*}
\end{prop}
\begin{IEEEproof}
See Appendix \ref{Appendixdiscretegeneral}.
\end{IEEEproof}
Although the result stated in Proposition~\ref{discretegen} seems intuitive and non-surprising, the proof is not straightforward, as a crucial step is the improvements proposed in Proposition~\ref{Sufprop}. Hence, if the channels $p_{Y|X}$ and $p_{Z|X}$ are symmetric, by Proposition~\ref{discretegen}, the auxiliary RV $U$ achieving $C_{\textup{rec}}(R_1,0)$ in Proposition~\ref{Sufprop}.\ref{C_rec3} also achieves $C_{\textup{WSK}}(R_1,0)$ in Proposition~\ref{Sufprop}.\ref{C_s_eq}; combined with Theorem \ref{theorem_seq1}, we obtain the following corollary.
\begin{cor}\label{cordiscretegen}
Let $(\mathcal{X}\mathcal{Y}\mathcal{Z},p_{XYZ})$ be a DMS such that $ X \text{---} Y\text{---} Z$ and $|\mathcal{X}|=2$. Let $R_1 \in \mathbb{R}_+^*$. We assume the channels $p_{Y|X}$ and $p_{Z|X}$ to be symmetric. The one-way rate-limited WSK capacity $C_{\textup{WSK}}(R_1,0)$ is achievable by a sequential strategy, moreover, reconciliation and privacy amplification steps can be handled independently, as defined in Section~\ref{SecIndep}.
\end{cor}
The following example illustrates Proposition~\ref{discretegen} and Corollary~\ref{cordiscretegen}.
\begin{ex}\label{Example_binary}
As depicted in Figure \ref{examplebscfig}, assume that $X$ and $Y$ (respectively $Y$ and $Z$) are connected by a binary symmetric channel (BSC) with crossover probability $p$ (respectively $q$). We also assume $X \sim \mathcal{B}(1/2)$ to obtain simpler expressions; however, the application of Proposition~\ref{discretegen} remains valid for $X \sim \mathcal{B}(\alpha)$, $\alpha \in [0,1]$.
\begin{figure}
\begin{center}
  \includegraphics[width=8.2cm]{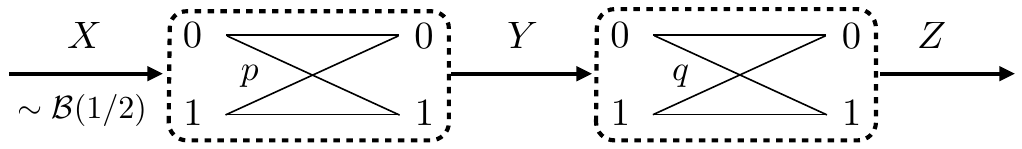}
  \caption{Binary DMS studied in Example \ref{Example_binary}}
\label{examplebscfig}
\vspace*{-1em}
\end{center}
\end{figure}
By Proposition \ref{discretegen}, the reconciliation capacity is \label{C_bin_rec}
$$
C_{\textup{rec}}(R_1,0)= 
 \begin{cases}
   1-{H}_b(p \star \beta_0),
 &\text{if } R_1 \leq {H}(X|Y), \\
   1-{H}_b(p), &\text{if } R_1 \geq {H}(X|Y),
   \end{cases}
$$
and the WSK capacity is \label{C_bin_s}
\begin{multline*}
C_{\textup{WSK}}(R_1,0) =\\
 \begin{cases}
   {H}_b\left( p\star \beta_0 \star q\right)-{H}_b(p \star \beta_0),
 & \hspace*{-0.1cm} \text{if } R_1 \leq {H}(X|Y), \\
   {H}_b(p \star q)  -{H}_b(p), &\hspace*{-0.1cm}\text{if } R_1 \geq {H}(X|Y),
   \end{cases}
\end{multline*}
with $\beta_0$, any of the two symmetric solutions of the equation ${H}_b(p \star \beta_0)-{H}_b(\beta_0)=R_1$.

Figure \ref{keycapbin} (resp. Figure \ref{reccapbin}) illustrates Remark~\ref{treshold_bin} and the fact that the reconciliation capacity $C_{\textup{rec}}(R_1,0)$ (resp. the secret key-capacity $C_{\textup{WSK}}(R_1,0)$) is monotonically increasing in the communication rate constraint.
\begin{figure}[H]
\input{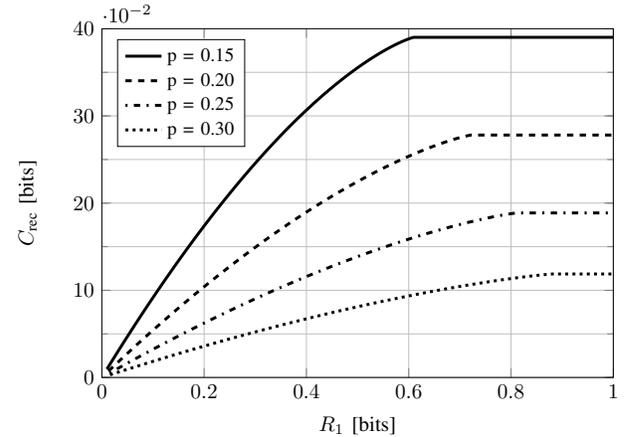}
  \caption{Reconciliation capacity $C_{\textup{rec}}(R_1,0)$} \label{reccapbin}
  \vspace*{-0.7em}
\end{figure}
\begin{figure}[H]
\input{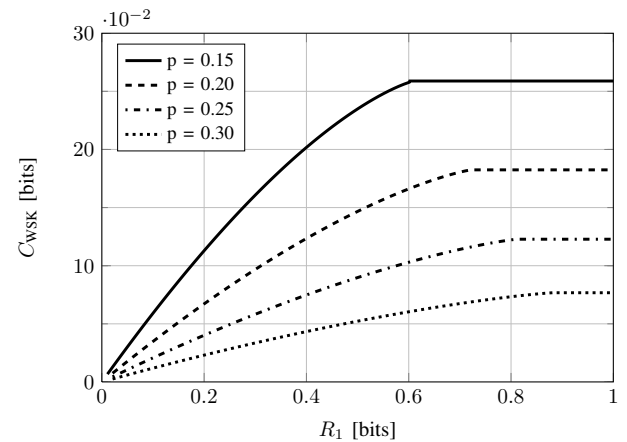}
  \caption{WSK capacity $C_{\textup{WSK}}(R_1,0)$ $(q=0.2)$}\label{keycapbin}
  \vspace*{-0.2em}
\end{figure}
Corollary \ref{cordiscretegen} states that choosing a test-channel $p_{U|X}$ as a BSC with parameter $\beta_0$, achieves $C_{\textup{rec}}(R_1,0)$ and $C_{\textup{WSK}}(R_1,0)$, so that reconciliation and privacy amplification can be designed independently. Consequently, for any other channel $p_{Z|Y}$, as long as $p_{Z|X}$ stays symmetric, the reconciliation capacity and the optimal reconciliation protocol for sequential key-generation remains the same. It is for instance the case if  we choose $p_{Z|Y}$ as a binary erasure channel (BEC), as depicted in Figure~\ref{examplebecfig}. Moreover, in this case, Proposition~\ref{discretegen} still allows us to determine the WSK capacity:
$$C_{\textup{WSK}}^{(\text{erasure})}(R_1,0) = 
 \begin{cases}
  \epsilon (1-{H}_b(p \star \beta_0)),
 & \hspace*{-0.1cm} \text{if } R_1 \leq {H}(X|Y), \\
   \epsilon (1 - {H}_b(p)), &\hspace*{-0.1cm}\text{if } R_1 \geq {H}(X|Y),
   \end{cases}
$$ 
where $\epsilon$ is the erasure probability characterizing $p_{Z|Y}$.
\end{ex}
\begin{figure}
\begin{center}
  \includegraphics[width=8.2cm]{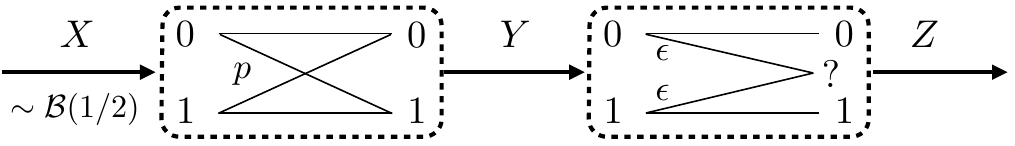}
  \caption{Binary erasure channel studied in Example \ref{Example_binary}}
\label{examplebecfig}
\end{center}
\end{figure}
\begin{rem}
We can show that the sequential strategy used in this section can also be applied to similar models for biometric secrecy~\cite{Ignatenko09}.
\end{rem}
\subsection{One-Way Rate-Limited WSK Capacity for Degraded Gaussian Sources} \label{Example_gauss}
In this section, we consider a degraded Gaussian MS with one-way rate-limited public communication. We assume that $X$, $Y$, and $Z$ are zero-mean correlated Gaussian sources on $\mathbb{R}$, and that Alice, Bob, and Eve know the covariance matrix of $(X,Y,Z)$.
We first refine the reconciliation capacity and the secret-key capacity to give the counterpart of Proposition~\ref{Sufprop}. We then provide the reconciliation capacity and the secret-key capacity, and show that reconciliation and privacy amplification can be treated independently. We also briefly discuss the performance of vector quantization compared to scalar quantization for the reconciliation step, thereby providing a counterpart of Remark~\ref{treshold_bin}.
\begin{prop}  \label{Prop_gauseq}
Let $(\mathcal{X}\mathcal{Y}\mathcal{Z},p_{XYZ})$ be a zero-mean Gaussian MS such that $ X \text{---} Y\text{---} Z$. Assume $R_1 \in \mathbb{R}^+$ and $R_2=0$. %Assume that the pdf $f_{U|X}$ exists and is in $\mathcal{B}_c(K)$, for some $K \in \mathbb{R}$.
  \begin{enumerate} [(a)]
 \item \label{C_recg} The one-way rate-limited reconciliation capacity is 
\begin{align*}
& C_{\textup{rec}}(R_1,0)  =  \displaystyle\max_{U} {I}(Y;U)  %\text{ subject to}
\end{align*}
\vspace*{-1em}
\text{ subject to }
\vspace*{-1em}
\begin{align} \label{ratecstrsigma}
&R_1 = {I}(X;U|Y),   \\
& U \text{---} X \text{---} Y, \nonumber.
\end{align}
\item  \label{C_s_eq_C}
The one-way rate-limited WSK capacity is 
\begin{align*}
& C_{\textup{WSK}}(R_1,0) = \displaystyle\max_{U} \left( {I}(Y;U) - {I}(Z;U)\right) %\text{ subject to }
\end{align*}
\vspace*{-1em}
\text{ subject to }
\vspace*{-1em}
\begin{align*} 
& R_1 = {I}(X;U|Y), \label{eqeq} \\ \nonumber
&  U \text{---} X \text{---} Y \text{---} Z, 
\end{align*}
\end{enumerate}
\end{prop}
Proposition \ref{Prop_gauseq} follows from Proposition~\ref{gaus_s_rec}.
%
%\begin{IEEEproof}
%See Appendix \ref{AppendixCMS}.
%\end{IEEEproof}
%
%
 
%
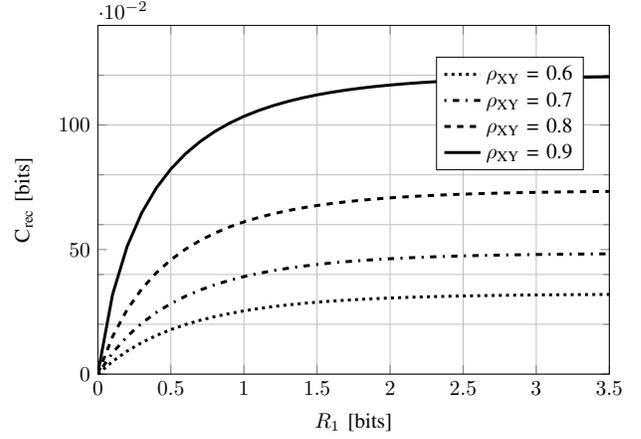
\begin{figure}%[H]
% This file was created by matlab2tikz v0.4.3.
% Copyright (c) 2008--2013, Nico Schlömer <nico.schloemer@gmail.com>
% All rights reserved.
% 
% The latest updates can be retrieved from
%   http://www.mathworks.com/matlabcentral/fileexchange/22022-matlab2tikz
% where you can also make suggestions and rate matlab2tikz.
% 
%
% defining custom colors
\definecolor{mycolor1}{rgb}{0,1,1}%
\begin{tikzpicture}[scale=0.8]

\begin{axis}[%
width=8.5cm,
height=5.8cm,
scale only axis,
xmin=0,
xmax=3.5,
xlabel={$R_1$ [bits]},
xmajorgrids,
ymin=0,
ymax=1.4,minor ytick={0,0.2,...,1.4},scaled y ticks={base 10:2},
ylabel={$\text{C}_{\textup{rec}}$ [bits]},
yminorgrids,
legend style={at={(0.95,0.91)}},
legend style={draw=black,fill=white,legend cell align=left}
]
\addplot [
color=black,
dotted,
line width=1.4pt
]
table[row sep=crcr]{
0.001 0.000561891386219815\\
0.101 0.0511568521334921\\
0.201 0.0924862328834469\\
0.301 0.126633120864082\\
0.401 0.155098008887653\\
0.501 0.1789953290735\\
0.601 0.199173457042643\\
0.701 0.216291401209003\\
0.801 0.230869701287935\\
0.901 0.24332530738204\\
1.001 0.25399614590694\\
1.101 0.263158838186715\\
1.201 0.271041749182474\\
1.301 0.277834775277374\\
1.401 0.283696806693318\\
1.501 0.288761500253605\\
1.601 0.293141803441969\\
1.701 0.296933541349791\\
1.801 0.300218290441529\\
1.901 0.303065702572071\\
2.001 0.305535400242336\\
2.101 0.307678533840003\\
2.201 0.309539069765511\\
2.301 0.31115486235363\\
2.401 0.312558550655065\\
2.501 0.313778312266328\\
2.601 0.314838499673627\\
2.701 0.31576017943336\\
2.801 0.31656159053904\\
2.901 0.317258535227678\\
3.001 0.317864713043787\\
3.101 0.318392007048771\\
3.201 0.3188507295212\\
3.301 0.319249833252035\\
3.401 0.319597093532653\\
3.501 0.319899265112544\\
3.601 0.320162217729636\\
3.701 0.320391053259762\\
3.801 0.320590207069963\\
3.901 0.320763535775064\\
};
\addlegendentry{$\rho{}_{\text{XY}}\text{ = 0.6}$};

\addplot [
color=black,
dash pattern=on 1pt off 3pt on 3pt off 3pt,
line width=1.4pt
]
table[row sep=crcr]{
0.001 0.000959480260467604\\
0.101 0.0853033432355531\\
0.201 0.151469307365225\\
0.301 0.204508714842193\\
0.401 0.247705905083769\\
0.501 0.283313904296278\\
0.601 0.312942894209072\\
0.701 0.337781293428929\\
0.801 0.358729205653238\\
0.901 0.376482898475939\\
1.001 0.391590430098763\\
1.101 0.404489512964862\\
1.201 0.415534027645105\\
1.301 0.425013050157126\\
1.401 0.433164802505676\\
1.501 0.440187075958658\\
1.601 0.446245150316838\\
1.701 0.451477901054124\\
1.801 0.456002572116561\\
1.901 0.459918550654126\\
2.001 0.463310384474258\\
2.101 0.466250217368815\\
2.201 0.468799771576564\\
2.301 0.471011974059112\\
2.401 0.472932299799076\\
2.501 0.474599888200601\\
2.601 0.476048476016199\\
2.701 0.47730718076417\\
2.801 0.478401161453135\\
2.901 0.479352177973822\\
3.001 0.480179066312462\\
3.101 0.480898143468188\\
3.201 0.481523553389148\\
3.301 0.482067563210035\\
3.401 0.482540817453043\\
3.501 0.482952556551809\\
3.601 0.483310805003855\\
3.701 0.483622533598259\\
3.801 0.483893799461278\\
3.901 0.484129867082093\\
};
\addlegendentry{$\rho{}_{\text{XY}}\text{ = 0.7}$};

\addplot [
color=black,
dashed,
line width=1.4pt
]
table[row sep=crcr]{
0.001 0.00177436201768549\\
0.101 0.150662908132666\\
0.201 0.259187312142113\\
0.301 0.34196350311635\\
0.401 0.407019681109543\\
0.501 0.459239028882568\\
0.601 0.501810852939345\\
0.701 0.536930404259429\\
0.801 0.566170707448617\\
0.901 0.590695191300222\\
1.001 0.611386529782709\\
1.101 0.628928461349317\\
1.201 0.643859806377652\\
1.301 0.656611317686163\\
1.401 0.6675315346869\\
1.501 0.676905368349756\\
1.601 0.684967747460972\\
1.701 0.691913827811611\\
1.801 0.697906757775453\\
1.901 0.703083673101892\\
2.001 0.70756038623663\\
2.101 0.711435098097894\\
2.201 0.714791367403945\\
2.301 0.717700508751322\\
2.401 0.720223545917832\\
2.501 0.722412815072101\\
2.601 0.724313289649855\\
2.701 0.725963681912187\\
2.801 0.727397363817899\\
2.901 0.728643140578519\\
3.001 0.729725903259874\\
3.101 0.730667181442907\\
3.201 0.731485612828849\\
3.301 0.732197343460678\\
3.401 0.732816369709587\\
3.501 0.733354831177275\\
3.601 0.733823262070374\\
3.701 0.734230807321343\\
3.801 0.734585408692064\\
3.901 0.734893965250269\\
};
\addlegendentry{$\rho{}_{\text{XY}}\text{ = 0.8}$};

\addplot [
color=black,
solid,
line width=1.4pt
]
table[row sep=crcr]{
0.001 0.00424767341285309\\
0.101 0.319386682179716\\
0.201 0.513139484813807\\
0.301 0.647690792818793\\
0.401 0.747338035781948\\
0.501 0.824125005113968\\
0.601 0.884903074391308\\
0.701 0.933940534208474\\
0.801 0.974073980860522\\
0.901 1.00728178854104\\
1.001 1.03499624811633\\
1.101 1.05828545140714\\
1.201 1.077965179294\\
1.301 1.09467084392497\\
1.401 1.10890548296324\\
1.501 1.12107279312977\\
1.601 1.13150048461712\\
1.701 1.14045718164321\\
1.801 1.14816490498898\\
1.901 1.15480845922836\\
2.001 1.16054260617204\\
2.101 1.16549762545409\\
2.201 1.16978368030062\\
2.301 1.17349428466286\\
2.401 1.1767090850837\\
2.501 1.17949611336797\\
2.601 1.18191362582754\\
2.701 1.1840116161006\\
2.801 1.18583306771635\\
2.901 1.18741499730308\\
3.001 1.18878932800296\\
3.101 1.18998362414974\\
3.201 1.19102171181241\\
3.301 1.19192420486345\\
3.401 1.19270895240602\\
3.501 1.19339142040897\\
3.601 1.19398501804861\\
3.701 1.19450137739079\\
3.801 1.19495059355463\\
3.901 1.19534143129698\\
};
\addlegendentry{$\rho{}_{\text{XY}}\text{ = 0.9}$};

\end{axis}
\end{tikzpicture}%
  \caption{ Reconciliation capacity $C_{\textup{rec}}(R_1,0)$ for different correlation coefficients $\rho_{XY}$}
\label{examplegauss}
\end{figure}

\begin{figure}%[H]
 % This file was created by matlab2tikz v0.4.3.
% Copyright (c) 2008--2013, Nico Schlömer <nico.schloemer@gmail.com>
% All rights reserved.
% 
% The latest updates can be retrieved from
%   http://www.mathworks.com/matlabcentral/fileexchange/22022-matlab2tikz
% where you can also make suggestions and rate matlab2tikz.
% 
%
% defining custom colors
\definecolor{mycolor1}{rgb}{0,1,1}%
\begin{tikzpicture}[scale=0.8]

\begin{axis}[%
width=8.5cm,
height=5.8cm,
scale only axis,
xmin=0,
xmax=5,
xlabel={$R_1$ [bits]},
xmajorgrids,
ymin=0,
ymax=0.175,minor ytick={0,0.025,...,0.17}, scaled y ticks={base 10:2},
ylabel={$C_{\textup{WSK}}$ [bits]},
yminorgrids,
legend style={at={(0.95,0.83)}},
legend style={draw=black,fill=white,legend cell align=left}
]
\addplot [
color=black,
dotted,
line width=1.4pt
]
table[row sep=crcr]{
0.001 1.21628377342531e-05\\
0.101 0.00114621613626387\\
0.201 0.0021320173263919\\
0.301 0.00298911157419163\\
0.401 0.00373442717878908\\
0.501 0.00438263568693176\\
0.601 0.00494646009583849\\
0.701 0.00543693918865962\\
0.801 0.00586365467022371\\
0.901 0.00623492665858888\\
1.001 0.00655798218333686\\
1.101 0.00683910060111542\\
1.201 0.0070837392292698\\
1.301 0.00729664199361979\\
1.401 0.00748193346637377\\
1.501 0.00764320031894172\\
1.601 0.00778356191948026\\
1.701 0.0079057315563687\\
1.801 0.00801206955846129\\
1.901 0.00810462940441575\\
2.001 0.0081851977614123\\
2.101 0.00825532926386668\\
2.201 0.008316376731764\\
2.301 0.00836951743312303\\
2.401 0.00841577591338877\\
2.501 0.00845604384427973\\
2.601 0.00849109728405517\\
2.701 0.00852161168894255\\
2.801 0.0085481749703562\\
2.901 0.00857129885354191\\
3.001 0.00859142875954644\\
3.101 0.00860895240319454\\
3.201 0.00862420727444352\\
3.301 0.00863748714854035\\
3.401 0.00864904775136238\\
3.501 0.00865911168981165\\
3.601 0.00866787274277864\\
3.701 0.00867549959573838\\
3.801 0.00868213909122229\\
3.901 0.00868791905799893\\
4.001 0.00869295077362381\\
4.101 0.00869733110791832\\
4.201 0.00870114438874858\\
4.301 0.00870446402611064\\
4.401 0.00870735392584759\\
4.501 0.00870986972026488\\
4.601 0.00871205983936807\\
4.701 0.00871396644337302\\
4.801 0.00871562623446015\\
4.901 0.00871707116341686\\
5.001 0.00871832904477726\\
5.101 0.0087194240923183\\
5.201 0.00872037738521841\\
5.301 0.00872120727386377\\
5.401 0.0087219297331142\\
5.501 0.00872255866983251\\
5.601 0.00872310619060029\\
5.701 0.00872358283477498\\
5.801 0.00872399777737315\\
5.901 0.00872435900569141\\
};
\addlegendentry{$\rho{}_{\text{XY}}\text{ = 0.6}$};

\addplot [
color=black,
dash pattern=on 1pt off 3pt on 3pt off 3pt,
line width=1.4pt
]
table[row sep=crcr]{
0.001 5.04937947464756e-05\\
0.101 0.00474675205725322\\
0.201 0.00881033557977212\\
0.301 0.0123293475684313\\
0.401 0.0153789099144062\\
0.501 0.0180232515422534\\
0.601 0.0203174183057139\\
0.701 0.0223086829662225\\
0.801 0.0240377152411711\\
0.901 0.0255395582462007\\
1.001 0.0268444474730756\\
1.101 0.0279785007716052\\
1.201 0.0289643019617334\\
1.301 0.0298213962095331\\
1.401 0.0305667118141305\\
1.501 0.0312149203222731\\
1.601 0.03177874473118\\
1.701 0.0322692238240009\\
1.801 0.032695939305565\\
1.901 0.0330672112939304\\
2.001 0.0333902668186782\\
2.101 0.0336713852364569\\
2.201 0.0339160238646113\\
2.301 0.0341289266289613\\
2.401 0.0343142181017152\\
2.501 0.0344754849542832\\
2.601 0.0346158465548216\\
2.701 0.0347380161917102\\
2.801 0.0348443541938028\\
2.901 0.034936914039757\\
3.001 0.0350174823967538\\
3.101 0.0350876138992081\\
3.201 0.0351486613671054\\
3.301 0.0352018020684643\\
3.401 0.0352480605487303\\
3.501 0.0352883284796212\\
3.601 0.0353233819193967\\
3.701 0.035353896324284\\
3.801 0.0353804596056976\\
3.901 0.0354035834888834\\
4.001 0.0354237133948878\\
4.101 0.0354412370385359\\
4.201 0.0354564919097851\\
4.301 0.0354697717838818\\
4.401 0.0354813323867039\\
4.501 0.0354913963251532\\
4.601 0.03550015737812\\
4.701 0.0355077842310798\\
4.801 0.0355144237265638\\
4.901 0.0355202036933404\\
5.001 0.0355252354089651\\
5.101 0.0355296157432597\\
5.201 0.0355334290240901\\
5.301 0.0355367486614521\\
5.401 0.035539638561189\\
5.501 0.0355421543556063\\
5.601 0.0355443444747095\\
5.701 0.0355462510787144\\
5.801 0.0355479108698017\\
5.901 0.0355493557987582\\
};
\addlegendentry{$\rho{}_{\text{XY}}\text{ = 0.7}$};

\addplot [
color=black,
dashed,
line width=1.4pt
]
table[row sep=crcr]{
0.001 0.000121264948924893\\
0.101 0.0113481279216468\\
0.201 0.0209814049803934\\
0.301 0.0292641754344202\\
0.401 0.0363981137632949\\
0.501 0.0425516374960069\\
0.601 0.0478661812033438\\
0.701 0.0524610930997657\\
0.801 0.0564375030221733\\
0.901 0.0598814110820875\\
1.001 0.0628661780400499\\
1.101 0.0654545508186472\\
1.201 0.0677003228024527\\
1.301 0.0696497042869144\\
1.401 0.0713424607371724\\
1.501 0.072812863455395\\
1.601 0.0740904875034142\\
1.701 0.0752008843672963\\
1.801 0.076166151237924\\
1.901 0.0770054144597983\\
2.001 0.0777352413412968\\
2.101 0.0783699918859649\\
2.201 0.0789221199220806\\
2.301 0.0794024314481934\\
2.401 0.0798203066797544\\
2.501 0.0801838912042623\\
2.601 0.0805002607749556\\
2.701 0.0807755635542323\\
2.801 0.0810151430255812\\
2.901 0.0812236443019352\\
3.001 0.0814051061495654\\
3.101 0.0815630407046008\\
3.201 0.0817005025718789\\
3.301 0.0818201487534224\\
3.401 0.0819242906486469\\
3.501 0.0820149391941646\\
3.601 0.0820938440626627\\
3.701 0.0821625277136532\\
3.801 0.0822223149804688\\
3.901 0.0822743587849248\\
4.001 0.0823196624911893\\
4.101 0.0823590993417018\\
4.201 0.0823934293587479\\
4.301 0.082423314044222\\
4.401 0.0824493291659737\\
4.501 0.0824719758809867\\
4.601 0.0824916904126204\\
4.701 0.0825088524705561\\
4.801 0.0825237925773154\\
4.901 0.0825367984437357\\
5.001 0.0825481205171503\\
5.101 0.0825579768098501\\
5.201 0.0825665571013616\\
5.301 0.0825740265958726\\
5.401 0.0825805291055518\\
5.501 0.0825861898212907\\
5.601 0.0825911177243975\\
5.701 0.0825954076858129\\
5.801 0.0825991422933669\\
5.901 0.0826023934423341\\
};
\addlegendentry{$\rho{}_{\text{XY}}\text{ = 0.8}$};

\addplot [
color=black,
solid,
line width=1.4pt
]
table[row sep=crcr]{
0.001 0.000238032711062819\\
0.101 0.0221109936971877\\
0.201 0.0406270287763458\\
0.301 0.0563680498463351\\
0.401 0.0697971883594569\\
0.501 0.0812877762515332\\
0.601 0.0911440198044849\\
0.701 0.0996160782095921\\
0.801 0.106911277654904\\
0.901 0.113202596492305\\
1.001 0.118635184849988\\
1.101 0.123331443112494\\
1.201 0.127395026635013\\
1.301 0.130914038623672\\
1.401 0.133963600969647\\
1.501 0.136607942597494\\
1.601 0.138902109360955\\
1.701 0.140893374021464\\
1.801 0.142622406296412\\
1.901 0.144124249301442\\
2.001 0.145429138528317\\
2.101 0.146563191826846\\
2.201 0.147548993016974\\
2.301 0.148406087264774\\
2.401 0.149151402869372\\
2.501 0.149799611377514\\
2.601 0.150363435786421\\
2.701 0.150853914879242\\
2.801 0.151280630360806\\
2.901 0.151651902349171\\
3.001 0.151974957873919\\
3.101 0.152256076291698\\
3.201 0.152500714919852\\
3.301 0.152713617684202\\
3.401 0.152898909156956\\
3.501 0.153060176009524\\
3.601 0.153200537610063\\
3.701 0.153322707246951\\
3.801 0.153429045249044\\
3.901 0.153521605094998\\
4.001 0.153602173451995\\
4.101 0.153672304954449\\
4.201 0.153733352422347\\
4.301 0.153786493123706\\
4.401 0.153832751603971\\
4.501 0.153873019534862\\
4.601 0.153908072974638\\
4.701 0.153938587379525\\
4.801 0.153965150660939\\
4.901 0.153988274544124\\
5.001 0.154008404450129\\
5.101 0.154025928093777\\
5.201 0.154041182965026\\
5.301 0.154054462839123\\
5.401 0.154066023441945\\
5.501 0.154076087380394\\
5.601 0.154084848433361\\
5.701 0.154092475286321\\
5.801 0.154099114781805\\
5.901 0.154104894748581\\
};
\addlegendentry{$\rho{}_{\text{XY}}\text{ = 0.9}$};

\end{axis}
\end{tikzpicture}%
    \caption{WSK capacity $C_{\textup{WSK}}(R_1,0)$, for different correlation coefficients $\rho_{XY}$ ($\rho_{XZ}=0.1$, $\rho_{YZ}=0.4$)}
\label{keycapgauss}
\end{figure}
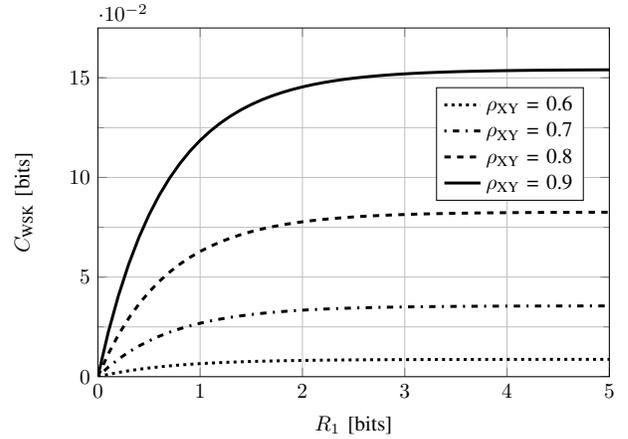
\begin{prop} \label{gaus_s_rec}
Assume that $(\mathcal{X}\mathcal{Y}\mathcal{Z},p_{XYZ})$ is a degraded zero-mean Gaussian source.  Let $R_1 \in \mathbb{R}_+$. 

The auxiliary RV $U$ achieving $C_{\textup{rec}}(R_1,0)$ in Proposition~\ref{Prop_gauseq}.\ref{C_recg} is a zero-mean Gaussian with variance $$\sigma_0 \triangleq \sigma_x \left( 1+ (1-\rho_{XY})(e^{2R_1}-1)^{-1}\right)$$ that satisfies the rate-constraint (\ref{ratecstrsigma}), where $\rho_{XY}$ is the correlation coefficient between $X$ and $Y$. Moreover, the same auxiliary RV $U$ achieves $C_{\textup{WSK}}(R_1,0)$ in Proposition~\ref{Prop_gauseq}.\ref{C_s_eq_C}.
\begin{enumerate}[(a)]
\item \label{C_gauss_rec} The one-way rate-limited reconciliation capacity is given by%
\begin{equation*}\label{region1}
C_{\textup{rec}}(R_1,0) = \frac{1}{2} \log_2  \frac{1 - \left( \rho_{XY} e^{-R_1}\right)^2}{1 - \rho_{XY}^2} .
\end{equation*}
\item \label{C_gauss_s} The one-way rate-limited WSK capacity is 
\begin{multline*}
\!\! \! C_{\textup{WSK}}(R_1,0)= \\   \!\! \!\!  \frac{1}{2} \log_2 \frac{  (1 - \rho_{YZ}^2)(1 - \rho_{XZ}^2) - \left(  \rho_{XY} - \rho_{YZ} \rho_{XZ} \right)^2e^{-2R_1} }{ (1 - \rho_{YZ}^2)(1 - \rho_{XZ}^2) - \left(  \rho_{XY} - \rho_{YZ} \rho_{XZ} \right)^2 } .
\end{multline*}
\end{enumerate}
\end{prop}
\begin{IEEEproof}
$(b)$ is due to Watanabe~\cite{Watanabe10a}, and the proof of $(a)$ is similar to the one of $(b)$.
\end{IEEEproof}
Proposition \ref{gaus_s_rec} states that both arguments of the maximum for the auxiliary RV $U$, in $(a)$ and $(b)$ of Proposition \ref{Prop_gauseq} are identical; combined with Theorem~\ref{theorem_seq1}, we deduce the following corollary.
\begin{cor}
Assume that $(\mathcal{X}\mathcal{Y}\mathcal{Z},p_{XYZ})$ is a degraded zero-mean Gaussian source. Let $R_1 \in \mathbb{R}_+$. %By Proposition~\ref{gaus_s_rec} and Theorem \ref{theorem_seq1}, t
The one-way rate-limited WSK capacity $C_{\textup{WSK}}(R_1,0)$ is achievable by a sequential strategy, moreover, reconciliation and privacy amplification steps can be handled independently, as defined in Section~\ref{SecIndep}.
\end{cor}
As shown by Proposition \ref{gaus_s_rec}.\ref{C_gauss_rec} (resp. Proposition \ref{gaus_s_rec}.\ref{C_gauss_s}), and as illustrated in Figure \ref{examplegauss} (resp. Figure \ref{keycapgauss}), the reconciliation  capacity (resp. the WSK capacity) does not reach ${I}(X;Y)$ (resp. ${I}(X;Y)-{I}(X;Z)$) when $R_1$ exceed a certain value. As mentioned in~\cite{Watanabe10a} and Remark \ref{treshold_bin}, unlike the case of discrete random variables, $C_{\textup{rec}}(R_1,0)$ (resp. $C_{\textup{WSK}}(R_1,0)$) can only approach ${I}(X;Y)$ (resp. ${I}(X;Y)-{I}(X;Z)$) asymptotically. Nevertheless, we show in the following proposition a continuous counterpart of Remark~\ref{treshold_bin}.
The achievability of $C_{\textup{WSK}}(R_1,0)$ with our sequential strategy is based on Wyner-Ziv coding. For a practical implementation, additional structure needs to be introduced, for instance with vector quantization. Since scalar quantization is the simplest and often the most computationally efficient type of quantization, it is natural to ask how scalar quantization performs compared to vector quantization. %
We answer this question in the following proposition.
\begin{prop} \label{Prop_Quant}
Let $n\in\mathbb{N}^*$, and $a>0$. Define $U$ as a uniformly quantized version of $X$. Specifically,%\\ $\forall k\in \llbracket 1, n\rrbracket$
\begin{equation*}
\forall k\in \llbracket 1, n\rrbracket, p_{U}(u_k)  \triangleq \int_{t_k}^{t_{k+1}} \! \!\!\!\!\!\! p_{X}(x) dx, \text{ with }t_k \triangleq a( 2\tfrac{ k -1}{n-1} - 1).
\end{equation*} 
%where $\forall k \in \llbracket 1,n \rrbracket, $.%, $t_0 \triangleq - \infty$, and $t_{n+1} \triangleq + \infty$.
%
%
%
If $n$ is large enough, then
\begin{equation*}
| {I}(X;Y) -  {I}(Y;U)|  \leq  \epsilon(a) +  a \cdot K e^{h(X|Y)- R_1},
\end{equation*}
where  $R_1$ is the communication rate constraint, $K$ is a constant, and $\epsilon(a)$ decreases exponentially fast to zero as $a$ goes to infinity.
\end{prop}
\begin{IEEEproof}
See Appendix \ref{AppendixQuant}.
\end{IEEEproof}
Proposition \ref{Prop_Quant} gives a continuous counterpart of Remark~\ref{treshold_bin}. Indeed, when $R_1 > h(X|Y)$, by Proposition~\ref{Prop_Quant}, if $X$ is quantized finely enough, then ${I}(Y;U)$ approach ${I}(X;Y)$ exponentially fast as $R_1$ increases.

Hence, the improvement of vector quantization compared to scalar quantization decays rapidly as the communication rate increases beyond $h(X|Y)$. Note that, in practice, we can optimize the scalar quantization, so that the loss could be even smaller than predicted by Proposition \ref{Prop_Quant}. Figure \ref{Quant_fig} illustrates this point by comparing the reconciliation capacity with numerical values of achievable rates obtained when $X$ is scalar-quantized.\footnote{We have increased the number of interval of quantization of $X$ from $2$ to $15$ and chosen their bounds by a standard gradient method to maximize ${I}(X_Q;Y)$.} Nevertheless, for low communication rates, Figure \ref{Quant_fig} shows that vector quantization improves the performance; in this case, we could implement, for instance, trellis coded vector quantization (TCVQ)~\cite{SWCTCQ}.
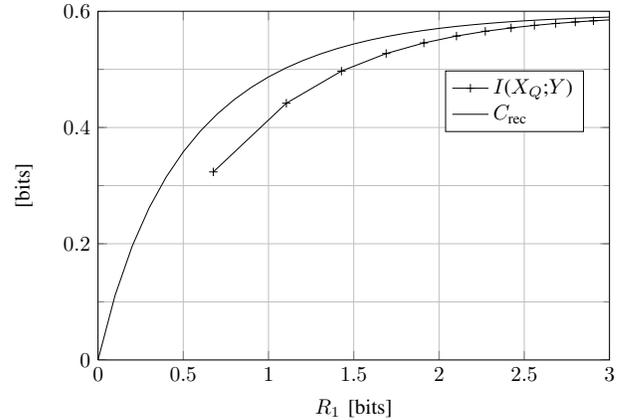
\begin{figure}%[H]
\begin{center}
  % This file was created by matlab2tikz v0.4.3.
% Copyright (c) 2008--2013, Nico Schlömer <nico.schloemer@gmail.com>
% All rights reserved.
% 
% The latest updates can be retrieved from
%   http://www.mathworks.com/matlabcentral/fileexchange/22022-matlab2tikz
% where you can also make suggestions and rate matlab2tikz.
% 
\begin{tikzpicture}[scale=0.8]

\begin{axis}[%
width=8.5cm,
height=5.8cm,
scale only axis,
xmin=0,
xmax=3,
xlabel={$R_1$ [bits]},
xmajorgrids,
ymin=0,
ymax=0.6, minor ytick={0,0.1,...,0.6},
ylabel={[bits]},
yminorgrids,
legend style={at={(0.95,0.83)}},
legend style={draw=black,fill=white,legend cell align=left}
]
\addplot [
color=black,
solid,
mark=+,
mark options={solid}
]
table[row sep=crcr]{
0.67638550506004 0.32361449493996\\
1.1046148958451 0.441784483490038\\
1.42795633648895 0.497070464467555\\
1.69011827409279 0.527247051263967\\
1.91124029668666 0.545495028970719\\
2.102688075949 0.55736251387163\\
2.27162861985128 0.565511295065019\\
2.42274182008936 0.571347231921303\\
2.55953439526967 0.575669561768667\\
2.68440381119964 0.578959886821946\\
2.79936781294508 0.581522441246926\\
2.90578633449379 0.583557073915484\\
3.00492906059734 0.585199456618566\\
3.09770539740512 0.586544316486328\\
};
\addlegendentry{$I$($X_Q$;$Y$)};
%$I(X_Q;Y)$
\addplot [
color=black,
solid
]
table[row sep=crcr]{
0 0\\
0.1 0.111052951282501\\
0.2 0.195512679792328\\
0.3 0.261760504122654\\
0.4 0.314860314138768\\
0.5 0.358103516999704\\
0.6 0.393747425008355\\
0.7 0.423404971018863\\
0.8 0.448266420969012\\
0.9 0.469233165346404\\
1 0.487002395733528\\
1.1 0.502122850969428\\
1.2 0.515032755875599\\
1.3 0.526086383981502\\
1.4 0.535573117555448\\
1.5 0.54373142062517\\
1.6 0.550759278057308\\
1.7 0.556822126163793\\
1.8 0.562058968123544\\
1.9 0.566587152931087\\
2 0.570506154763537\\
2.1 0.573900593975663\\
2.2 0.576842675166649\\
2.3 0.579394171786477\\
2.4 0.581608054107086\\
2.5 0.583529833873962\\
2.6 0.585198681796708\\
2.7 0.586648361360939\\
2.8 0.587908012969612\\
2.9 0.589002815263631\\
3 0.589954545007467\\
3.1 0.590782052713909\\
3.2 0.591501667905855\\
3.3 0.592127545342031\\
3.4 0.592671961499035\\
3.5 0.593145568979496\\
3.6 0.593557615212143\\
3.7 0.593916130754404\\
3.8 0.594228091648408\\
3.9 0.594499559576537\\
4 0.594735802981527\\
};
\addlegendentry{$C_{\textup{rec}}$};

\end{axis}
\end{tikzpicture}%
%\includegraphics[width=8cm]{Quant.pdf}
  %\vspace*{-1em}
  \caption{Reconciliation capacity obtained for a scalar quantization of $X$ with $\rho_{XY} = 0.75$, $h(X|Y) \approx 1$}
\label{Quant_fig}
\end{center}
\end{figure}
\section{Concluding remarks}
We have shown that the the best known bounds for the one-way rate-limited capacity are achievable by a sequential strategy that separates reliability and secrecy thanks to a reconciliation step followed by a privacy amplification step with extractors; in the case of two-way communication, the sequential design seems to suffer a loss of performance compared to the joint design and similar secret key rates have only been established for degraded sources or when there is no side information at the eavesdropper (SK capacity). We have also qualified robustness of sequential strategy to rate-limited communication, by showing that achieving the reconciliation capacity in a sequential strategy is, unlike the case of rate-unlimited communication, not necessarily optimal. We further provide scenarios for which it stays optimal. % for rate-limited communication.
As a side result, we have extended known bounds of the WSK capacity for a discrete source model to the case of a continuous source model, and derive a closed-form expression of the one-way rate-limited capacity for degraded binary symmetric sources.

A strength of sequential key-generation is to easily translate into practical designs. Even more interestingly, the proposed scheme can be made very flexible with the following modifications.
\subsubsection{Rate-compatible reconciliation}
we can adapt to the characteristics of the legitimate users by the use of rate-compatible LDPC codes, to perform the reconciliation phase, as demonstrated in \cite{Elkouss10,Kasai10}. Note, however, that vector quantization might be required, which could complexify the reconciliation~phase.
\subsubsection{Rate-compatible privacy amplification}
Privacy amplification can also be performed with universal families of hash functions, in which case the counterpart of Theorem~\ref{thMaurer} is found in~\cite{Bennett95}.\footnote{ However, it requires more random bits than extractors, on the order of $n$ random bits, since functions must be chosen at random in universal families. Consequently, our scheme needs to be adapted to account for it.} Hence, one can design privacy amplification methods easily adjustable to the characteristics of the eavesdropper's observations, if we make $k$ vary in the following universal family of hash functions $\mathcal{H} =\{ \text{GF}(2^n) \rightarrow \{ 0,1\}^k, x \mapsto (k \text{ bits of the product } xy) | y \in  \text{GF}(2^n) \}$, where the $k$ bits are fixed but their position can be chosen arbitrarily~\cite{Carter79}. 
\appendices
\section{Proof of Theorem \ref{theorem_seq2}}\label{AppendixThcont}
%In the following, we use the same notations as in Appendix \ref{AppendixC_rec}.
\subsubsection{Discrete case}
 Let $\epsilon > 0$. Let $R_1, R_2 \in \mathbb{R}^+$. Let $m,n \in \mathbb{N}$, and define $N \triangleq nm$. Let $k \in \mathbb{N}$ to be determined later. Consider a sequential key-distillation strategy $\mathcal{S}_{N}$ that consists~of 
\begin{itemize}
\item $m$ repetitions of a reconciliation protocol $\mathcal{R}_{n}$ based on Wyner-Ziv coding. % (
The protocol $\mathcal{R}_{n}$ operates as described in Appendix \ref{recach}. 
Hence, after one repetition of the protocol, Alice obtains $S^n \triangleq U^n\hat{V}^{n}$, whereas Bob has $\hat{S}^n \triangleq \hat{U}^nV^n$ and $\mathbb{P}[\hat{S}^n \neq S^{n}|\mathcal{R}_{n}] \leq P_e (\epsilon,n)$.\footnote{By Appendix~\ref{recach}, $P_e (\epsilon,n)$ decreases exponentially to zero as $n\epsilon^2$ goes to infinity.} In addition, the information disclosed over the public channel during the $m$ repetitions of the reconciliation protocol is upper bounded by $\log  |\mathcal{A}|^m + \log |\mathcal{B}|^m = N ( {I}(U;X) - {I}(U;Y) + {I}(V;Y|U) - {I}(V;X|U) +  r_0(\epsilon))$, with $\lim_{\epsilon \to 0} r_0(\epsilon) = 0$.\footnote{$r_0(\epsilon) \triangleq 6 \epsilon H(U) +12 \epsilon H(V|U)$ by Appendix~\ref{recach}.} An additional round of reconciliation is then performed to ensure $\mathbb{P}[ (\hat{S}^n)^m \neq (S^{n})^m|\mathcal{R}_{n}] \leq \delta_e (m)$, where $\lim_{m\rightarrow \infty}\delta_e (m) = 0$, for any fixed $n$. We note $\log |\mathcal{C}|^m$ the information communicated to perform this last step. Hence, the overall information disclosed is upper bounded by $l_{rec}\triangleq \log (|\mathcal{A}|^N |\mathcal{B}|^N|\mathcal{C}|^m)$, that is
\begin{align}
l_{rec} 
%&\triangleq \log (|\mathcal{A}|^N |\mathcal{B}|^N|\mathcal{C}|^m) \nonumber \\
& = N ( {I}(U;X) - {I}(U;Y)  \nonumber \\
& \phantom{mllll}+ {I}(V;Y|U)\! - {I}(V;X|U)+  r_1(\epsilon,n)), \label{lrec} \\
&\text{with }r_1(\epsilon,n) \triangleq \frac{1+ \epsilon }{n} H(S^{n}|\hat{S}^{n} ) + r_0(\epsilon) \label{rester1}
\end{align}
%with $\lim_{\epsilon_0\rightarrow 0}r_0'(\epsilon_0) = 0$;+ m (H(S^{n}|\hat{S}^{n} ) + r_0'(\epsilon_0))
%\begin{equation}\label{rester1}
%\text{with }r_1(\epsilon,n) \triangleq \frac{1+ \epsilon }{n} H(S^{n}|\hat{S}^{n} ) + r_0(\epsilon)
%\end{equation}
arbitrarily small for $n$ large enough by Fano's inequality, so that the communication rates $R_1$ and $R_2$ remain asymptotically unchanged.
\item privacy amplification, based on extractors with output size $k$, at the end of which Alice computes her key $K\triangleq g(S^{N},U_d)$, while Bob computes $\hat{K}\triangleq g(\hat{S}^{N},U_d)$, where $U_d$ is a sequence of $d$ uniformly distributed random bits.
\end{itemize}
The total information available to Eve after reconciliation consists of her observation $Z^{N}$, the public messages $A^m$ and $B^m$, respectively sent by Alice and Bob, the public message $C^m$, and $U_d$. The strategy $\mathcal{S}_{N}$ is also known to Eve, but we omit the conditioning on $\mathcal{S}_{N}$ for convenience. 

We first show that, for a suitable choice of the output size $k$, the quantity $ k -{H}(K|U_dZ^{N}A^mB^mC^m)$ vanishes to zero for $N$ large enough. Then, we show that the corresponding WSK rate achieves the lower bound on the WSK capacity of Theorem~\ref{C_WSK}.
We first state Lemma \ref{lem1}, a refined version of the results in \cite{Maurer00,Bloch11}, that is obtained by using the notion of robust typicality developed in the appendix of~\cite{Orlitsky01}, to later extend our result to the continuous case.
\begin{lem}[\cite{Maurer00,Bloch11}, Refined version] \label{lem1}
Consider a DMS $(\mathcal{X}\mathcal{Z},p_{XZ})$ and define the RV $\Theta$ as
\begin{align*}
\Theta & \triangleq
 \begin{cases}
   1 & \text{if } (X^{n} , Z^{n}) \in  \mathcal{T}_{2\epsilon}^{n}(XZ) \text{ and }  Z^{n} \in  \mathcal{T}_{\epsilon}^{n}(Z),\\
  0 & \text{otherwise.}  
   \end{cases}
\end{align*}
Then, $\mathbb{P}[\Theta=1]\geq1- \delta_{\epsilon}^0(n)$, with  $\delta_{\epsilon}^0(n) \triangleq 2|S_X|e^{-\epsilon^2 n \mu_X /3}+2|S_{XZ}|e^{-\epsilon^2 n \mu_{XZ} /3}$, where $S_X \triangleq \{ x \in \mathcal{X} : p(x)>0 \}$ and $\mu_X \triangleq \min_{x \in S_X} p(x)$ . Moreover, if $z^n \in \mathcal{T}_{\epsilon}^{n}(Z)$,
\begin{multline*}
%{H}_c(X^n|Z^n=z^n, \Theta = 1) 
{H}_{\infty}(X^n|Z^n=z^n,\Theta=1) \\ \geq n({H}(X|Z)-\delta^0(\epsilon)) + \log( 1-\delta_{\epsilon}^1(n)), 
\end{multline*}
where $\delta^0(\epsilon) \triangleq \epsilon {H}(X|Z)$ and $\delta_{\epsilon}^1(n) \triangleq 2|S_{X,Z}|e^{-\epsilon^2 n\mu_{X,Z}/6}$.
\end{lem}
Let us start by defining the following RVs
\begin{align*}
\Theta & \triangleq 
   \begin{cases}
   1 & \text{if } (S^{N} , Z^{N}) \in  \mathcal{T}_{2\epsilon}^{m}(U^nV^nZ^n) \\
   & \phantom{mmmmmmmm}\text{ and }  Z^{N} \in  \mathcal{T}_{\epsilon}^{m}(Z^n), \\
   0 & \text{otherwise.}  
   \end{cases}\\
\Upsilon &  \triangleq 
   \begin{cases}
   1 & \text{if } {H}_{\infty}(S^{N} | z^{N},a^m, b^m,c^m,\Theta =1) \\
   & \phantom{m} \geq {H}_{\infty}(S^{N} | z^{N}, \Theta =1)- l_{rec} - \sqrt{N},\\
   0 & \text{otherwise.}  
   \end{cases}
\end{align*}
By Lemma \ref{lem1} applied to the DMS $(\mathcal{U}^n \mathcal{V}^n \mathcal{Z}^n,p_{U^nV^nZ^n})$,  $\mathbb{P}[\Theta=1] \geq 1-\delta_{\epsilon}^0(m)$, and by \cite[Lemma 10]{Maurer00}, $\mathbb{P}[\Upsilon=1] \geq 1-2^{-\sqrt{N}}$. Hence, $\mathbb{P}[\Upsilon=1,\Theta=1] \geq 1 - \delta_{\epsilon}^0(m) - 2^{-\sqrt{N}}$, %by the union bound
and
\begin{multline}
{H}(K|U_dZ^{N}A^mB^mC^m)  %\nonumber
%& \geq {H}(K |U_d Z^{N}A^r \Upsilon \Theta)\\ %\nonumber
%& \geq  \mathbb{P}(\Upsilon=1,\Theta=1) {H}(K |U_d Z^{N}A^r, \Upsilon=1, \Theta=1)\\ 
 \geq \left(1 - \delta_{\epsilon}^0(m) - 2^{-\sqrt{N}}\right) \\ \times {H}(K |U_d Z^{N}A^mB^mC^m, \Upsilon = 1, \Theta=1). 
  \label{eqdeb3}
\end{multline}
To lower bound ${H}(K |U_d Z^{N}A^mB^mC^m, \Upsilon=1, \Theta=1)$, we first lower bound ${H}_{\infty}(S^{N} | z^{N}, a^m, b^m,c^m, \Theta =1, \Upsilon =1)$ to be able to use Theorem \ref{thMaurer}. By definition of $\Upsilon$,
\begin{align}
 & {H}_{\infty}(S^{N} | z^{N},a^m,b^m,c^m,\Theta =1, \Upsilon =1)   \nonumber \\   \nonumber
& \stackrel{\phantom{}}{\geq} {H}_{\infty}(S^{N} | Z^{N}=z^{N}, \Theta =1) - l_{rec}- \sqrt{N}  \\   
&\stackrel{(a)}{\geq}  m({H}(S^n|Z^n)-n r_2(\epsilon,n,m) ) -  l_{rec},\label{eqdeb2}
\end{align}
where $(a)$ follows from Lemma \ref{lem1} with 
\begin{equation}\label{rester2}
r_2(\epsilon,n,m) \triangleq \epsilon \frac{H(S^n|Z^n)}{n}- N^{-1}\log( 1-\delta_{\epsilon}^1(m)) + N^{-1/2}.\footnote{The $m$ repetitions of the protocol $\mathcal{R}_n$ allow us to link ${H}_{\infty}(\cdot)$ to ${H}(\cdot)$.}
\end{equation}
We now lower bound ${H}(S^n|Z^n)$. %We first remark that
\begin{align}
& \nonumber  {H}(S^n|Z^n) \\ \nonumber
& =  {H}(\hat{S}^n|Z^n) + {H}(S^n|\hat{S}^nZ^n) - {H}(\hat{S}^n|S^nZ^n) \\ \nonumber
& \stackrel{(b)}{\geq}    {H}(\hat{S}^n|Z^n) -  \delta_{\epsilon}(n)\\ \nonumber
& =  {I}(Y^n;\hat{S}^n|Z^n)  + {H}(\hat{S}^n|Y^nZ^n) -  \delta_{\epsilon}(n)\\ \nonumber
& = {H}(Y^n|Z^n) - {H}(Y^n|Z^n\hat{S}^n)  + {H}(\hat{U}^n|Y^nZ^n) \\ \nonumber
& \phantom{mmmmmmmmmm} + {H}(V^n|Y^n\hat{U}^nZ^n) -  \delta_{\epsilon}(n)\\ %\nonumber
%& \stackrel{(b)}{=}  {H}(Y^n|Z^n) - {H}(Y^n|Z^n\hat{S}^n)  + {H}(\hat{U}^n|Y^nZ^n) + {H}(V^n|Y^n\hat{U}^n) -  \delta_{\epsilon}(n)\\ 
& \stackrel{(c)}{=}  n {H}(Y|Z)   - {H}(Y^n|Z^n\hat{S}^n) + {H}(\hat{U}^n|Y^nZ^n) -  \delta_{\epsilon}(n),\label{eqcomb0}
\end{align}
where $(b)$ follows from Fano's inequality where $\lim_{n \rightarrow \infty}\delta_{\epsilon}(n) =0 $ by the exponential decrease of $P_e (\epsilon,n)$ with $\epsilon^2 n$, and $(c)$ holds because $V^n$ is a function of $(Y^n\hat{U}^n)$, and the $Y_i$'s and $Z_i$'s are i.i.d..
We first lower bound ${H}(\hat{U}^n|Y^nZ^n)$. %Remark that\vspace*{-0.cm}
\begin{align}
&{H}(\hat{U}^n|Y^nZ^n)  \nonumber\\
&= {H}(U^n|Y^nZ^n)   + {H}(\hat{U}^n|U^nY^nZ^n) - {H}(U^n|\hat{U}^nY^nZ^n)  \nonumber\\
&\stackrel{(d)}{\geq}  {H}(U^n|Y^nZ^n) - \delta_{\epsilon}(n) \nonumber\\
%&=  {I}(X^n;U^n|Y^nZ^n) + {H}(U^n|X^nY^nZ^n) - \delta_{\epsilon}(n)\nonumber\\
&\geq {I}(X^n;U^n|Y^nZ^n)   - \delta_{\epsilon}(n) \nonumber\\ 
%&= {H}(X^n|Y^nZ^n) - {H}(X^n|Y^nZ^nU^n) - \delta_{\epsilon}(n) \nonumber\\
&\stackrel{\mathclap{\scriptstyle{(e)}}}{=} n {H}(X|YZ) - {H}(X^n|Y^nZ^nU^n) - \delta_{\epsilon}(n), \label{eqdeb}
\end{align}
\vspace*{-0.cm}where $(d)$ follows from Fano's inequality where $\lim_{n \rightarrow \infty}\delta_{\epsilon}(n) =0 $ by the exponential decrease of $P_e (\epsilon,n)$ with $\epsilon^2 n$, 
%$(d)$ holds because $U^n$ is a function of $X^n$ ($U^n$ is a quantized version of $X^n$), 
and $(e)$ holds since the $X_i$'s, $Y_i$'s , and $Z_i$'s are i.i.d.. Then, define 
\begin{align*}
\Gamma & \triangleq 
 \begin{cases}
   1 & \text{if } (X^{n} , U^{n}, Y^n,Z^n) \in  \mathcal{T}_{2\epsilon}^{n}(XUYZ),\\
  0 & \text{otherwise.}  
   \end{cases}\\
%\end{align*}
%\begin{align*}
\Delta  & \triangleq 
 \begin{cases}
   1 & \text{if } (X^{n} , U^{n}) \in  \mathcal{T}_{\epsilon}^{n}(XU),\\
  0 & \text{otherwise.}  
   \end{cases}
\end{align*}
 so that,
\begin{align}
&{H}(X^n|Y^nZ^nU^n) \nonumber \\ \nonumber
& \leq {H}(X^n \Gamma \Delta |Y^nZ^nU^n) \\ \nonumber
& = {H}(\Gamma \Delta | Y^n Z^n U^n) + {H}(X^n | Y^nZ^n U^n \Gamma \Delta)\\ \nonumber
& \leq 2 + \smash{ \sum_{\mathclap{\delta, \gamma \in \left\{0,1 \right\} } }} \ \mathbb{P}[\Gamma=\gamma | \Delta=\delta] \mathbb{P}[\Delta=\delta] \nonumber \\ \nonumber
&\phantom{mmmmmmm} \times {H}(X^n | Y^nZ^n U^n, \Gamma=\gamma, \Delta=\delta) \\ \nonumber
&  \smash{\stackrel{(f)}{\leq}} 2+  {H}(X^n | Y^nZ^n U^n, \Gamma=1, \Delta=1) \\ 
& \phantom{mmmmmmmm}+ n ( 2\delta_{\epsilon}^2(n)+ \delta_{\epsilon}^4(n) )\log |\mathcal{X}|, \label{equation1}
\end{align}
where $(f)$ holds since $\mathbb{P}[\Delta =0]\triangleq \delta_{\epsilon}^2(n)$,%
\footnote{We have $\delta_{\epsilon}^2(n) \leq P_e(\epsilon,n)$ by Appendix \ref{recach}.} and $\mathbb{P}[\Gamma=0 | \Delta=1]\leq \delta_{\epsilon}^4(n)$.\footnote{By Markov Lemma, we have $\delta_{\epsilon}^4(n) \triangleq 2|S_{UXYZ}|e^{-\epsilon^2 n\mu_{UXYZ}/6} $.} Indeed, we can apply Markov Lemma~\cite{Berger78} (see the version given in~\cite{Orlitsky01}), since we have $U^n \text{---} X^n \text{---} Y^nZ^n$ and for every $(x^n,y^n,z^n)$, $p(y^nz^n|x^n) = \displaystyle\prod_{i=1}^n p_{YZ|X}(y_iz_i|x_i)$. 
Then,
\begin{align}
& \nonumber   {H}(X^n | Y^n Z^n U^n, \Gamma=1, \Delta=1) \\ \nonumber
& = {\sum_{y^n, z^n, u^n}} p(y^n,z^n,u^n|1,1)  {H}(X^n | y^n, z^n, u^n, \Gamma \! = \! 1,  \Delta \! = \! 1) \\  \nonumber
%& = \displaystyle\sum_{z^n, u^n} p(z^n, u^n|1,1) \displaystyle\sum_{x^n : (x^n, z^n, u^n) \in E_0} p(x^n| z^n u^n) \log \frac{1}{p( x^n |z^n u^n)}\\
& \leq \sum_{y^n, z^n, u^n} p(y^n,z^n,u^n|1,1) \log |\mathcal{T}_{2\epsilon}^n(X| y^n, z^n, u^n)| \\  \nonumber
& \leq  \sum_{y^n, z^n, u^n} p(y^n,z^n,u^n|1,1) (n {H}(X|YZU) (1+ 2\epsilon))\\
& \leq n {H}(X|YZU) (1+ 2\epsilon). \label{equation2}
\end{align}
Hence, combining (\ref{eqdeb}), (\ref{equation1}), and (\ref{equation2}), we obtain
\begin{equation} \label{eqcomb1}
{H}(\hat{U}^n|Y^nZ^n) \geq n ({H}(X|YZ) - {H}(X|YZU) - r_3(\epsilon,n)),
\end{equation}
where 
\begin{multline} \label{reste1}
r_3(\epsilon,n) \triangleq 2{H}(X|YZU) \epsilon + ( 2\delta_{\epsilon}^2(n)+ \delta_{\epsilon}^4(n) )\log |\mathcal{X}| \\+ 2/n +  \delta_{\epsilon}(n)/n.
\end{multline}
We now lower bound  the term $-{H}(Y^n|Z^n\hat{S}^n)$ in (\ref{eqcomb0}). Define 
\begin{align*}
\Gamma_1 & \triangleq 
 \begin{cases}
   1 & \text{if } (Y^n,\hat{U}^{n} , V^{n}, Z^n) \in  \mathcal{T}_{2\epsilon}^{n}(YUVZ),\\
  0 & \text{otherwise.}  
   \end{cases}\\
%\end{align*}
%\begin{align*}
\Delta_1 \vspace*{-1.4em} & \triangleq 
 \begin{cases}
   1 & \text{if } (Y^n, \hat{U}^{n} , V^{n}) \in  \mathcal{T}_{\epsilon}^{n}(YUV),\\
  0 & \text{otherwise.\vspace*{-1.4em}}  
   \end{cases}\vspace*{-1.4em}
\end{align*}
We can write
\begin{align}
& {H}(Y^n|Z^n\hat{S}^n) \nonumber \\ \nonumber
& \leq {H}(Y^n \Gamma_1 \Delta_1 |Z^n\hat{S}^n)\\ \nonumber
& = {H}(\Gamma_1 \Delta_1 | Z^n \hat{S}^n) + {H}(Y^n | Z^n \hat{S}^n \Gamma_1 \Delta_1)\\ \nonumber
& \leq 2 +  \smash{\sum_{\delta_1, \gamma_1 \in \left\{0,1 \right\} }}\mathbb{P}[\Gamma_1=\gamma_1 | \Delta_1=\delta_1] \mathbb{P}[\Delta_1=\delta_1]  \\ \nonumber
& \phantom{pyjamammmmmm} \times {H}(Y^n | Z^n \hat{S}^n, \Gamma_1=\gamma_1, \Delta_1=\delta_1)\\ 
& \smash{\stackrel{(g)}{\leq}} 2+  {H}(Y^n|Z^n\hat{S}^n, \Gamma_1=1, \Delta_1=1)  \nonumber \\
& \phantom{pyjamammmmmmm} + n ( 2\delta_{\epsilon}^3(n)+ \delta_{\epsilon}^5(n) )\log |\mathcal{Y}|, \label{eqaux1}
\end{align}
where $(g)$ holds since $\mathbb{P}[\Delta_1 =0]\triangleq \delta_{\epsilon}^3(n)$,%
\footnote{We have $\delta_{\epsilon}^3(n) \leq P_e(\epsilon,n)$ by Appendix \ref{recach}.}
  and $\mathbb{P}[\Gamma_1=0 | \Delta_1=1]\leq \delta_{\epsilon}^5(n)$.\footnote{By Markov Lemma, we have $\delta_{\epsilon}^5(n) \triangleq 2|S_{UVYZ}|e^{-\epsilon^2 n\mu_{UVYZ}/6} $.} Indeed, we can apply Markov Lemma, since we have for every $(y^n,z^n)$, $p(z^n|y^n) = \displaystyle\prod_{i=1}^n p_{Z|Y}(z_i|y_i)$, and $(\hat{U}^nV^n) \text{---} Y^n \text{---} Z^n$, which follows from the assumption $X\text{---} Y \text{---} Z$.\footnote{Note that the assumption of degraded sources is only necessary here. The use of this hypothesis is the weakness, at least for two-way communication (for one-way communication this assumption is not necessary), of a proof that consists of a successive design of reconciliation and privacy amplification, rather than a joint design as in \cite{Csiszar00}, where the joint design is exploited to get the joint typicality of $(V^n,Y^n,\hat{U}^n,Z^n)$. } 
\begin{align}
&  {H}(Y^n|Z^n\hat{S}^n, \Gamma_1=1, \Delta_1=1) \nonumber \\ \nonumber
& = \displaystyle\sum_{z^n, \hat{s}^n} p(z^n, \hat{s}^n|1,1)  {H}(Y^n | z^n, \hat{s}^n, \Gamma_1 = 1,  \Delta_1 = 1) \\ \nonumber
%& = \displaystyle\sum_{z^n, u^n} p(z^n, u^n|1,1) \displaystyle\sum_{x^n : (x^n, z^n, u^n) \in E_0} p(x^n| z^n u^n) \log \frac{1}{p( x^n |z^n u^n)}\\
& \leq \displaystyle\sum_{z^n, \hat{s}^n} p(z^n, \hat{s}^n|1,1)  \log |\mathcal{T}_{2\epsilon}^n(Y| z^n,\hat{s}^n)| \\ \nonumber
& \leq \displaystyle\sum_{z^n, \hat{s}^n} p(z^n, \hat{s}^n|1,1)  (n {H}(Y|ZUV) (1+ 2\epsilon))\\
& \leq n {H}(Y|ZUV) (1+ 2\epsilon).  \label{eqaux2}
\end{align}
Hence by (\ref{eqaux1}), (\ref{eqaux2}), 
\begin{align} \label{eqcomb2}
 {H}(Y^n|Z^nU^nV^n) \leq n ({H}(Y|ZUV) + r_4(\epsilon,n)),
\end{align}
where
\begin{equation} \label{reste2}
r_4(\epsilon,n) \triangleq  2{H}(Y|ZUV) \epsilon+ ( 2\delta_{\epsilon}^3(n)+ \delta_{\epsilon}^5(n) )\log |\mathcal{Y}| +2/n.
\end{equation}
Combining (\ref{eqcomb0}), (\ref{eqcomb1}), (\ref{eqcomb2}),
\begin{multline}
  {H}(S^n|Z^n)  \geq n[{H}(Y|Z) +{H}(X|YZ)   - {H}(X|YZU) \\ -{H}(Y|ZUV) - r_3(\epsilon,n) - r_4(\epsilon,n)] - \delta_{\epsilon}(n). \label{eqdebcomb}
\end{multline}
Then, remark that
\begin{align}
& {H}(Y|Z) +{H}(X|YZ)  - {H}(X|YZU) -{H}(Y|ZUV)  \nonumber  \\ \nonumber
&= {I}(Y;UV|Z) + {I}(X;U|YZ)   \\  \nonumber
%&=  {H}(UV|Z) - {H}(UV|YZ) +  {I}(X;U|YZ)  \\ \nonumber
&=  {H}(U|Z) + {H}(V|UZ)  - {H}(V|UYZ) -  {H}(U|XYZ) \\ \nonumber
&\overset{(h)}{\geq} {H}(U|Z) + {H}(V|UZ)  - {H}(V|UY) -  {H}(U|X) \\ %\nonumber
%&= {H}(U|Z) - {I}(V;Z|U) -{H}(U|YZ) + {I}(V;Y|U) +  {I}(X;U|YZ)\\ \nonumber
%&\overset{(i)}{\geq} {H}(U|Z) - {I}(V;Z|U) -{H}(U|YZ) + {I}(V;Y|U) -  {H}(U|X) +  {H}(U|YZ)\\ 
&= {I}(U;X) - {I}(U;Z) - {I}(V;Z|U) + {I}(V;Y|U), \label{eqsup}
\end{align}
where $(h)$ holds because conditioning reduces entropy.
Hence, by (\ref{lrec}), (\ref{eqdeb2}), (\ref{eqdebcomb}), and (\ref{eqsup})
\begin{multline}
 {H}_{\infty}(S^{N} | z^{N}, a^m,b^m,c^m, \Theta =1, \Upsilon =1)     \\
 \geq N[{I}(U;Y) + {I}(V;X|U) - {I}(U;Z)  - {I}(V;Z|U)  \\
 -r_5(\epsilon,n,m)] , \label{setk}
\end{multline}
\vspace*{-0.4em}
where
\begin{multline}\label{reste3}
r_5(\epsilon,n,m)  \triangleq r_1(\epsilon,n) + r_2(\epsilon,n,m) \\ +r_3(\epsilon,n) +r_4(\epsilon,n)+\delta_{\epsilon}(n)/n. 
%& =  -2N ({H}(X|YZU) +{H(Y|Z\hat{S})} ) \epsilon - N (2\delta_{\epsilon}^0(n)+ \delta_{\epsilon}^1(n)/2 ) \log ( |\mathcal{X}| |\mathcal{Y}|) -4m + m \delta_{\epsilon}(n)
\end{multline}
Set $k$ to be less than the lower bound in (\ref{setk}) by $\sqrt{N}$: 
\begin{multline} \label{equationk}
k \triangleq \lfloor N[{I}(U;Y) + {I}(V;X|U) - {I}(U;Z) - {I}(V;Z|U) \\ - r_5(\epsilon,N)] - \sqrt{N} \rfloor.
\end{multline}
Now with (\ref{setk}) and (\ref{equationk}), we can apply  Theorem \ref{thMaurer} to lower bound ${H}(K |U_d Z^{N}A^mB^mC^m, \Upsilon=1, \Theta=1)$ by $k -\delta^*(N)$, where $\delta^*(N) = 2^{-\sqrt{N}/\log N} \left( k+\sqrt{N}/\log N\right)$. Thus, we can finally lower bound ${H}(K|U_dZ^{N}A^mB^mC^m)$ in (\ref{eqdeb3}):
\begin{align*}
&{H}(K|U_dZ^{N}A^mB^mC^m) \\
& \geq \left(1 - \delta_{\epsilon}^0(m) - 2^{-\sqrt{N}} \right) (k - \delta^*(N)) \\
& = k - r_6(\epsilon,n,m),
\end{align*}
\vspace*{-0.4em}
where 
\begin{multline*}
r_6(\epsilon,n,m) \triangleq  \left(1 - \delta_{\epsilon}^0(m) - 2^{-\sqrt{N}} \right) \delta^*(N) \\ + \left(  \delta_{\epsilon}^0(m) +2^{-\sqrt{N}} \right) k.
\end{multline*}
Moreover, the leakage is such that 
\begin{multline} \label{leakage}
{I}(K;U_dZ^{N}A^mB^mC^m) \\ = {H}(K)-{H}(K|U_dZ^{N}A^mB^mC^m) \leq r_6(\epsilon,n,m),
\end{multline}
with $r_6(\epsilon,n,m)$ vanishing to zero for a fixed $n$ as $m$ goes to infinity.
The keys computed by Alice and Bob are asymptotically the same for a fixed $n$ as $m$ goes to infinity, since
\begin{equation}\label{kreliability}
\mathbb{P}[K \neq \hat{K}] \leq \mathbb{P}[ (S^n)^m \neq (\hat{S}^n)^m] \leq \delta_e (m).  %m \mathbb{P}((U^n\hat{V}^n) \neq (\hat{U}^n V^n)) \leq m P_e(\epsilon, n).
\end{equation}

Then, by (\ref{rester1}), (\ref{rester2}), (\ref{reste1}), (\ref{reste2}), (\ref{reste3}), we have that $r_5(\epsilon,n,m)$ vanishes to zero for $n$ large enough and as $m$ goes to infinity, thus the secret key rate $R \triangleq k/N$ is asymptotically as close as desired to
\begin{equation*}
{I}(U;Y) - {I}(U;Z) + {I}(V;X|U)  - {I}(V;Z|U).
\end{equation*} 
Note that it is not exactly the bound proposed in Theorem~\ref{C_WSK}.\ref{C_WSK2} for the WSK capacity. We finish the proof as follows. If ${I}(V;X|U)  \leq {I}(V;Z|U)$, in the reconciliation we set $R_2=0$ so that the asymptotic secret key rate is now as close~as~desired~to
\begin{equation*}
 {I}(U;Y) - {I}(U;Z) + [{I}(V;X|U)  - {I}(V;Z|U)]^+.
\end{equation*}
Then, if ${I}(U;Y) \leq {I}(U;Z)$, in the reconciliation protocol, we choose $S^n=V^n$  (see the beginning of the proof), and we assume that $U^N$ is provided by a genie to Eve. Consequently, we obtain instead of Equation (\ref{eqdeb2}),
\begin{multline*}
  {H}_{\infty}(V^{N} | z^{N}, u^N, b^m,c^m,\Theta =1, \Upsilon =1)   \\   
\geq  m({H}(V^n|Z^nU^n)-n r_2(\epsilon,n,m))  \\ -  N ({I}(V;Y|U) - {I}(V;X|U) - r_1(\epsilon,n)),
\end{multline*}
and conclude in the same manner, to obtain an asymptotic secret key rate as close as desired to
\begin{equation*}
[ {I}(U;Y) - {I}(U;Z) ]^+ + [{I}(V;X|U)  - {I}(V;Z|U)]^+ .
\end{equation*}
\subsubsection{Continuous case}
We use the following lemma to extend the result to the continuous case by means of quantization.
\begin{lem}[\cite{Fano61,Pinsker64,Cover91}] \label{lemiquant}
Let $X$ and $Y$ be two real-valued random variables with probability distribution $\mathbb{P}_X$ and $\mathbb{P}_Y$ respectively. Let $\mathcal{E}_{\Delta_1}=\left\{ E_i\right\}_{i \in \mathcal{I}}$, $\mathcal{F}_{\Delta_2}=\left\{ F_j\right\}_{j \in \mathcal{J}}$ be two partitions of $X$ and $Y$ such that for any $i \in \mathcal{I}, \mathbb{P}_X(E_i)=\Delta_1$, for any $j \in \mathcal{J}, \mathbb{P}_Y(F_j)=\Delta_2$, where $\Delta_1, \Delta_2 >0$. Let $X_{\Delta_1}$, $Y_{\Delta_2}$ be the quantized version of $X$, $Y$ with respect to the partitions $\mathcal{E}_{\Delta_1}$, $\mathcal{F}_{\Delta_2}$ respectively. Then, we have
$$
{I}(X;Y) = \lim_{\Delta_1, \Delta_2 \to 0} {I}(X_{\Delta_1},Y_{\Delta_2}).
$$\end{lem}
Note that a quantization of  the eavesdropper observation $Z^n$ might underestimate its knowledge from the legitimate users point of view and implicitly increase the leakage. However, by Lemma \ref{lemiquant}, for any $\delta >0$, if the quantized version $Z^n_{\Delta^n}$ of $Z^n$ is fine enough,  then the leakage is not compromised and
\begin{equation*}
|I(K;MZ^n) -I(K;MZ^n_{\Delta^n})| < \delta.
\end{equation*}
This argument is also used in~\cite{barros08,Chou13c,Pierrot13}.

We perform the quantization as follows. As in Lemma~\ref{lemiquant}, we jointly quantify $X$, $Y$, $Z$, $U$ and $V$ to form $X_{\Delta_X}$, $Y_{\Delta_Y}$, $Z_{\Delta_Z},U_{\Delta_U}$, $V_{\Delta_V}$ such that $\Delta_X=\Delta_Y=\Delta_Z=\Delta_U=\Delta_V=l^{-b}$ and $|\mathcal{X}_{\Delta_X}|=|\mathcal{Y}_{\Delta_Y}|=|\mathcal{Z}_{\Delta_Z}|=|\mathcal{U}_{\Delta_U}|=|\mathcal{V}_{\Delta_V}|=l^{b}$ with $b>0$. 
We now apply the proof of the discrete case to the random variables $X_{\Delta_X}$, $Y_{\Delta_Y}$, $Z_{\Delta_Z},U_{\Delta_U}$, $V_{\Delta_V}$. By Lemma~\ref{lemiquant}, we can fix $l$ large enough such that
$|{I}(U_{\Delta_U};Y_{\Delta_Y}) - {I}(U;Y)| < \delta/4$, $|{I}(V_{\Delta_V};X_{\Delta_X}|U_{\Delta_U}) - {I}(V;X|U)| < \delta/4$, $|{I}(U_{\Delta_U};Z_{\Delta_Z}) - {I}(U;Z)| < \delta/4$, $|{I}(V_{\Delta_V};Z_{\Delta_Z}|U_{\Delta_U}) - {I}(V;Z|U)| < \delta/4$, and Equation (\ref{equationk}) becomes 
\begin{multline*} 
k \geq \lfloor N[{I}(Y;U) -{I}(V;X|U) - {I}(U;Z) - {I}(V;Z|U) \\- r_5(\epsilon,n,m)-  \delta]  - \sqrt{N}  \rfloor.
\end{multline*}
At this point, we cannot conclude with the last inequality. Indeed, in the term $r_5(\epsilon,n,m)$ are hidden the following terms: ${H}(X_{\Delta_X}|Z Y_{\Delta_Y}U_{\Delta_U})\epsilon$ (see (\ref{reste1})),  ${H}(Y_{\Delta_Y}|Z_{\Delta_Z} U_{\Delta_U} V_{\Delta_V})\epsilon$ (see (\ref{reste2})), ${H}(U_{\Delta_U})\epsilon$ and ${H}(V_{\Delta_V}|U_{\Delta_U})\epsilon$ (by definition of $r_0(\epsilon)$), which do not vanish to $0$ as $l$ get large. Now, if we choose $\epsilon = n^{-a}$, where $a \in ] 0, 1/2[$, so that for $i\in \{ 0,1,2,3,5 \} $, $\delta_{\epsilon}^i(n)$ vanishes as $n$ get large for $l$ fixed,\footnote{Recall that $P_e (\epsilon,n)$ decreases exponentially to zero as $n \epsilon^2$ goes to infinity.} then the asymptotic secret-key rate, for $n$ large enough and as $m$ goes to infinity becomes as close as desired to
\begin{equation*}
  {I}(Y;U) -{I}(V;X|U) - {I}(U;Z) - {I}(V;Z|U).
\end{equation*} 
Moreover, the leakage in (\ref{leakage}), and the key error probability between Alice an Bob in (\ref{kreliability}), still vanish to zero for $n$ large enough and as $m$ goes to infinity.

\section{Proof of Theorem \ref{theorem_seq1}} \label{AppendixTh1}
As in \cite{Csiszar04}, Theorem \ref{theorem_seq1} is not directly deduced from Theorem \ref{theorem_seq2}. % 
 We first consider the case of one-way public communication, in which Alice sends messages to Bob, a first time with rate $R_1'$ and a second time with rate $R_2'$. For this scenario we note $C_{\textup{rec}}^*$ the reconciliation capacity. 

We can modify the proof of Proposition \ref{C_rec2} to obtain for $R_1',R_2' \in \mathbb{R}^+$, 
\begin{align*}
& C_{\textup{rec}}^*(R_1',R_2') \geq \displaystyle\max_{U,V} \left[{I}(U;Y)  + {I}(V;Y|U) \right] %\text{ subject to}
\end{align*}
\vspace*{-1em}
subject to
\vspace*{-1em}
\begin{align}
&R_1' \geq {I}(X;U|Y) \label{rate1b} \\
&R_2' \geq {I}(V;X|YU) \label{rate2b} \\
&U \text{---} X \text{---} Y, \text{ } V \text{---} UX \text{---} Y. \nonumber
\end{align}
Then, we can modify the proof of Theorem \ref{theorem_seq2} to prove that we can achieve the rate
\begin{multline*}
 R_{\textup{WSK}}^*(R_1',R_2') \triangleq \displaystyle\max_{U,V} \left( [ {I}(Y;U) -{I}(Z;U) ]^+  \right. \\ \left. + [ {I}(Y;V|U)  - {I}(Z;V|U) ] ^+ \right),
\end{multline*}
subject to rate constraints (\ref{rate1b}), (\ref{rate2b}) and Markov conditions
\begin{align}
&U \text{---} X \text{---} YZ, \text{ }V \text{---} UX \text{---} YZ \label{markov2b},
\end{align}
by a reconciliation phase followed by a privacy amplification phase performed with extractors, and this time without the assumption $X \rightarrow Y \rightarrow Z$. 
Note that Markov condition 
\begin{align}
U \text{---} V \text{---} X \text{---} YZ \label{markov3},
\end{align}
 implies Markov conditions (\ref{markov2b}), and that if Markov condition~(\ref{markov3}) holds, then the rate constraint~(\ref{rate2b}) becomes
\begin{multline*} 
R_2' \geq  {I}(X;V|U) - {I}(Y;V|U) 
 \geq  {I}(X;V) -{I}(Y;V) - R_1'. 
\end{multline*}
Hence, for $R_1',R_2'>0$ satisfying $R_1' +R_2' = R_1$,
%\begin{align*}
%R_{\textup{WSK}}^*(R_1',R_2') \geq \displaystyle\max_{U,V}  [{I}(Y;V|U) - {I}(Z;V|U) ].
%\end{align*}
%subject to Markov conditions  (\ref{markov1b}), (\ref{markov2b}), and rate constraints ${I}(X;U|Y) \leq R_1'$, ${I}(V;X|YU) \leq R_2'$, where $R_1',R_2'>0$ satisfy $R_1' +R_2' = R_2$.
%
%so that
\begin{align*} 
&  R_{\textup{WSK}}^*(R_1',R_2') \geq \displaystyle\max_{U,V}  [{I}(Y;V|U) - {I}(Z;V|U) ],
\end{align*}
subject to rate constraint $R_1 \geq {I}(X;V) -{I}(Y;V)$ and Markov condition (\ref{markov3}). We conclude by observing that $C_{\textup{WSK}}(R_1,0) \geq R_{\textup{WSK}}^*(R_1',R_2')$. % Hence, $ R_{\textup{WSK}}^*(0,R_2) \geq C_{WSK}(R_2,0)$ .
\section{Proof of  Proposition \ref{C_rec2}} \label{AppendixC_rec}
\subsection{Converse}
Let $R_1, R_2 \in \mathbb{R}^+$. We first establish the rate constraints on $R_1$ and $R_2$. We have
\begin{align}
n R_1 \nonumber
&\geq {H}(A)  \\ \nonumber
%&\geq  {I}(A;X^n)    \\ \nonumber
&\geq  {I}(A;X^n)  - {I}(A;Y^n)   \\ \nonumber
&\stackrel{(a)}{=} n [ {I}(A;X_J|\tilde{U}) - {I}(A;Y_J|\tilde{U}) ]\\ %\nonumber
&\stackrel{(b)}{=} n [ {I}(U;X_J) - {I}(U;Y_J) ],  \label{eqrate1aux} %\\
%&\stackrel{(h)}{=} n  {I}(U;X_J|Y_J)
\end{align}
where $(a)$ holds by \cite[Lemma 4.1]{Ahlswede98}, if we set $\tilde{U}\triangleq X^{J-1}Y_{J+1}^NJ$ and $J$ is a RV uniformly distributed on $\llbracket 1, n\rrbracket$, independent of all previous RVs, $(b)$ holds if we set $U \triangleq A\tilde{U}$, since $X_J$ and $\tilde{U}$ are independent.\\% and (h) holds since $U \rightarrow X_J \rightarrow Y_J$ forms a Makov chain.\\
Similarly, we have
\begin{align}
 nR_2 \nonumber
&\stackrel{\phantom{(a)}}{\geq} {H}(B|A)  \\ \nonumber
&\stackrel{(c)}{\geq}  {H}(B|X^n) + {H} (\hat{S}|S)  - n\delta(\epsilon) \\ \nonumber
&\stackrel{(d)}{\geq} {I} (\hat{S};B|X^n) + {H}(\hat{S}|BX^n) - n\delta(\epsilon)\\ 
&= {H}(\hat{S}|X^n) \label{eqrate2aux} - n\delta(\epsilon) \\ \nonumber
%&=  {H}(\hat{S}|X^nA) \\ \nonumber
&= {H}(\hat{S}|A) - {I}(\hat{S};X^n|A) - n\delta(\epsilon) \\ \nonumber
&\stackrel{(e)}{=} {I}(\hat{S};Y^n|A) - {I}(\hat{S};X^n|A) - n\delta(\epsilon) \\ \nonumber
&\stackrel{(f)}{=}  n [ {I}(V;Y_J|U) - {I}(V;X_J|U)] - n\delta(\epsilon), %\\ \nonumber
%&\stackrel{(e)}{=}  n {I}(UV;Y_J|X_J) - n\delta(\epsilon),
\end{align}
where $(c)$ holds because $A$ is a function of $X^n$ and by Fano's inequality, since for any $\epsilon >0$, there exists a reconciliation protocol such that $\mathbb{P} [ S \neq \hat{S} ] \leq \delta(\epsilon)$,\footnote{$\delta(\epsilon)$ denotes a function of $\epsilon$ such that $\lim_{\epsilon \rightarrow 0}\delta(\epsilon) =0$.} $(d)$ holds since $S=\eta_a(X^n,B)$, $(e)$ holds since $\hat{S} = \eta_b(Y^n,A)$, $(f)$ holds by \cite[Lemma 4.1]{Ahlswede98} and if we set $V \triangleq \hat{S}$. % finally (e) holds since  $V \rightarrow Y_JU \rightarrow X_J$ and $U \rightarrow X_J \rightarrow Y_J$.\\

We now determine the reconciliation capacity bound.
\begin{align}
 {I}(\hat{S};X^n) \nonumber
& = \displaystyle\sum_{i=1}^{n} {I}(\hat{S};X_i|X^{i-1}) \\ \nonumber
&\stackrel{(a)}{=}  \displaystyle\sum_{i=1}^{n} {I}(\hat{S}X^{i-1};X_i) \\ \nonumber
& \leq  \displaystyle\sum_{i=1}^{n} {I}(\hat{S}X^{i-1}Y_{i+1}^n;X_i) \\ \nonumber
&=  n \displaystyle\sum_{i=1}^{n} \mathbb{P}[J=i] {I}(\hat{S}X^{J-1}Y_{J+1}^{n};X_J| J=i) \\ \nonumber
&=  n {I}(\hat{S}\tilde{U};X_J|J) \\
&\leq  n {I}(VU;X_J), \label{eqaux}
\end{align}
where $(a)$ holds because the $X_i$'s are i.i.d.. Then, 
\begin{align*}
&{H}(\hat{S}) - {H}(AB)\\
&= {I}(\hat{S};X^n) + {H}(\hat{S}|X^n) - {H}(A)  - {H}(B|A) \\
&\stackrel{(b)}{\leq} n {I}(VU;X_J) - {H}(A) + n\delta(\epsilon) \\
&\stackrel{(c)}{\leq} n [ {I}(V;X_J|U) +I(U;Y_J) +  \delta(\epsilon)],%\\
%& = n [{I}(X_J;Y_J)-{I}(X_J;Y_J|UV) +  \delta(\epsilon)],
\end{align*}
where $(b)$ holds by (\ref{eqaux}) and since ${H}(\hat{S}|X^n) \leq {H}(B|A)+ n\delta(\epsilon) $ by (\ref{eqrate2aux}), and $(c)$ holds by (\ref{eqrate1aux}).\\
For a DMS, standard techniques \cite{Ahlswede98} show that $|\mathcal{U}| \leq |\mathcal{X}|+2$ and $|\mathcal{V}| \leq |\mathcal{Y}|$.
\subsection{Achievability} \label{recach}
The proof for a DMS is similar to Wyner-Ziv coding~\cite{Wyner76}, we only describe the protocol. In the following, for $n \in \mathbb{N}$ and $\epsilon>0$,  we note $\mathcal{T}_{\epsilon}^n(X)$ the set of $\epsilon$-letter-typical sequences~\cite{Kramerbook} (also called ``robustly typical sequence" in \cite{Orlitsky01}) with respect to $p_X$. We also define conditional typical sets as follows, $\mathcal{T}_{\epsilon}^n(Y|x^n) \triangleq \{ y^n : (x^n,y^n) \in \mathcal{T}_{\epsilon}^n(XY) \}$. We note $\mu_X \triangleq \min_{x\in supp(p_X)}p_X(x)$. 
Let $\epsilon >0$, and define $\epsilon_1 \triangleq \frac{1}{2}\epsilon$, $\epsilon_2 \triangleq 2 \epsilon$.\\
\textbf{Code construction}: 
Fix a joint probability distribution $p_{UX}$ on $\mathcal{U} \times \mathcal{X}$ and $p_{UVY}$ on $\mathcal{U} \times \mathcal{V} \times \mathcal{Y}$. Let $R_u^{\phantom{a}} = {I}(X;U|Y) + 6 \epsilon H(U)$, $R_u'={I}(Y;U) - 3 \epsilon H(U)$. Generate $2^{n(R_u^{\phantom{a}}+R_u')}$ codewords, labeled $u^n(\omega,\nu)$ with $(\omega,\nu) \in \llbracket 1, 2^{nR_u^{\phantom{a}}\phantom{'}} \rrbracket \times \llbracket 1, 2^{nR_u'}\rrbracket$, by generating the symbols $u_i(\omega,\nu)$ for $i \in \llbracket 1,n\rrbracket$ and $(\omega,\nu) \in \llbracket 1, 2^{nR_u^{\phantom{a}}\phantom{'}} \rrbracket \times \llbracket 1, 2^{nR_u'}\rrbracket$ independently according to $p_U$. Let $R_v^{\phantom{a}} = {I}(V;Y|XU) + 6 \epsilon_2 H(V|U)$, $R_v'={I}(V;X|U) - 3 \epsilon_2 H(V|U)$. For each $(\omega,\nu)$, generate $2^{n(R_v^{\phantom{a}}+R_v')}$ codewords, labeled $v^n(\omega,\nu,k,l)$ with $(k,l) \in \llbracket 1, 2^{nR_v^{\phantom{a}}\phantom{'}}\rrbracket \times \llbracket 1, 2^{nR_v'}\rrbracket$, by generating the symbols $v_i(\omega,\nu,k,l)$ for $i \in \llbracket 1,n\rrbracket$ and  $(k,l) \in \llbracket 1, 2^{nR_v^{\phantom{a}}\phantom{'}}\rrbracket \times \llbracket 1, 2^{nR_v'}\rrbracket$ independently according to $p_{V|U=u_i(\omega,\nu)}$.\\ 
\textbf{Step1. At Alice's side}: 
Given $x^n$, find a pair $(\omega,\nu)$ s.t $(x^n,u^n(\omega,\nu)) \in \mathcal{T}^n_{\epsilon}(XU)$. If we find several pairs, we choose the smallest one (by lexicographic order). If we fail we choose $(\omega,\nu)=(1,1)$.
Define $s_1^n \triangleq u^n(\omega,\nu)$ and transmit $a \triangleq \omega$.\\
\textbf{Step2. At Bob's side}:
Given $y^n$ and $a$, find $\tilde{\nu}$ s.t $(y^n,u^n(\omega,\tilde{\nu})) \in \mathcal{T}^n_{\epsilon}(YU)$ and define $\hat{s}_1^n \triangleq u^n(\omega,\tilde{\nu})$. If there is one or more such $\tilde{\nu}$, choose the lowest, otherwise set $\tilde{\nu}=1$. Find a pair $(k,l)$ such that $\left( \hat{s}_1^n,y^n,v^n(\omega,\tilde{\nu},k,l) \right) \in \mathcal{T}^n_{\epsilon_2}(UYV)$. If there is one or more such $(k,l)$, choose the lowest, otherwise set $(k,l)=(1,1)$. Transmit $b=k$. Define $\hat{s}_2^n \triangleq v^n(\omega,\tilde{\nu},k,l) $ and $\hat{s}^n \triangleq (\hat{s}_1^n,\hat{s}_2^n)$.\\
\textbf{Step3. At Alice's side}: 
Given $s_1^n=u^n(\omega,\nu)$ and $b$, find $\tilde{l}$ s.t $( x^n,s_1^n,v^n(\omega,\tilde{\nu},k,\tilde{l}) ) \in \mathcal{T}^n_{\epsilon_2}(XUV)$. If there is one or more such $\tilde{l}$, choose the lowest, otherwise set $\tilde{l}=1$.
Define $s_2^n \triangleq v^n(\omega,\tilde{\nu},k,\tilde{l})$ and $s^n \triangleq (s_1^n,s_2^n)$.\\
We can show by standard arguments that there exists a code, such that after one repetition of the protocol, Alice obtains $S^n = U^n\hat{V}^{n}$, whereas Bob has $\hat{S}^n = \hat{U}^nV^n$ with $\mathbb{P}[\hat{U}^n \neq U^{n}] \leq \delta_{\epsilon}(n)$,\footnote{$\delta_{\epsilon}(n)$ denotes a function of $\epsilon$  and $n$ such that $\lim_{n \rightarrow \infty}\delta_{\epsilon}(n) =0$.} $\mathbb{P}[\hat{V}^n \neq V^{n}] \leq \delta_{\epsilon}(n)$, $\mathbb{P} [\hat{S}^n \neq S^{n}|\mathcal{R}_{n}] \leq P_e (\epsilon,n)$\footnote{We can show that $P_e (\epsilon,n)$ decreases exponentially to zero as $n\epsilon^2$ goes to infinity.} and $(U^n,X^n)$, $(\hat{U}^n,Y^n)$, $(\hat{U}^n,Y^n,V^n)$, $(U^n,\hat{V}^n,X^n)$ jointly typical with probability approaching one for $n$ large.\\
To extend the result to a CMS, we proceed as in the proof of Theorem \ref{theorem_seq2}.
\section{Proof of Proposition \ref{Sufprop}} \label{AppendixC_rec1}
\subsection{Proof of Part i)} 
The achievability and converse proof can be found in \cite{Chou12}, it remains to prove that equality holds in the rate constraint~(\ref{rate1}) and that $|\mathcal{U}|\leq |\mathcal{X}|$.
\subsubsection{Equality constraint}\label{defdelta0}
We start with the following lemma.
\begin{lem} \label{fconvex}
$f(U) \triangleq {I}(Y;U)$ and $f_1(U) \triangleq {I}(X;U|Y)$ are convex in $p_{U|X}$.
\end{lem}
\begin{proof}
Let $\lambda \in [0,1]$, let $U_1$, $U_2$ defined by $p_{U_1|X}$ and $p_{U_2|X}$ respectively, be s.t.  $U_1 \text{---} X \text{---} Y$ and $U_2 \text{---} X \text{---} Y$.\\ We introduce the random variable $Q\in \left\{ 1,2\right\}$ independent of all others and set $U = U_Q$.
\begin{align*}
Q & \triangleq 
 \begin{cases}
   1 & \text{ with probability }  \lambda,\\
   2 &  \text{ with probability } 1-\lambda.
   \end{cases}
\end{align*}
\begin{align*}
{I}(Y;U) 
& \leq  {I}(Y;UQ) \\
& = {I}(Y;Q) + {I}(Y;U|Q)\\
& \stackrel{(a)}{=} {I}(Y;U|Q)\\
& = \lambda {I}(Y;U_1) + (1-\lambda) {I}(Y;U_2),
\end{align*}
where $(a)$ holds since $Y$ and $Q$ are independent.
\begin{align*}
{I}(X;U|Y)
& \leq  {I}(X;UQ|Y)\\
& = {I}(X;Q|Y) + {I}(X;U|YQ) \\
& \stackrel{(b)}{=}  {I}(X;U|YQ)\\
& = \lambda ({I}(X;U_1|Y) + (1 - \lambda) ({I}(X;U_2|Y),
\end{align*}
where $(b)$ holds because ${H}(X|YQ) = {H}(X|Y)$, since $Q$ and ($X,Y$) are independent.
\end{proof}
\label{defdelta}
By Lemma \ref{fconvex}, $f(U)$ and $f_1(U)$ are convex in $p_{U|X}$.
Define $\Delta \triangleq \{ \textbf{u}  \in \mathbb{R}^{|\mathcal{U}||\mathcal{X}|} :  \forall i,j\in \llbracket 1 , |\mathcal{U}|\rrbracket \times \llbracket 1,|\mathcal{X}|\rrbracket, \sum_{k=1}^{|\mathcal{U}|} u_{kj} =1,  u_{ij} \geq 0\}$, and
$\mathcal{C} \triangleq \left\{ \textbf{u} \in \Delta : f_1(\textbf{u}) \leq R_1 \right\} $.\\ We first show that $\mathcal{C}$ is convex compact, with extreme points in $\left\{ \textbf{u} \in \Delta : f_1(\textbf{u}) = R_1 \right\} $:
\begin{itemize}
\item $\mathcal{C}$ is the preimage of $[0 , R_1]$ by the continuous function $f_1$, thus $\mathcal{C}$ is closed. We deduce that $\mathcal{C}$ is compact, since $\mathcal{C} \subset [0,1]^{|\mathcal{U}||\mathcal{X}|}$ and $[0,1]^{|\mathcal{U}||\mathcal{X}|}$ is compact.
\item $\mathcal{C}$ is convex by convexity of $f_1$, since the sublevels of a convex function are convex sets.
\item Let $\textbf{u}_1 \in \mathcal{C}$ s.t. $f_1(\textbf{u}_1) = R_1 - \delta$, with $\delta >0$. 
By continuity of $f_1$, $\exists \epsilon_0, \forall \textbf{u} \in \mathcal{B}(\textbf{u}_1,\epsilon_0), | f_1(\textbf{u})-f_1(\textbf{u}_1)| < \delta$. Let $\textbf{u}_0 \in \mathcal{B}(\textbf{u}_1,\epsilon_0) \backslash \left\{\textbf{u}_1\right\}$, $\lambda \in \left\{ -\frac{1}{2} , +\frac{1}{2} \right\}$ and $\textbf{u}_{\lambda} =\lambda \textbf{u}_0 + (1-\lambda) \textbf{u}_1 $. \\
Then $|| \textbf{u}_{\lambda} - \textbf{u}_1 ||  = || \lambda (\textbf{u}_0 -\textbf{u}_1) || \leq |\lambda| \epsilon_0 $, which means $\textbf{u}_{\lambda} \in \mathcal{C}$. Hence, $\frac{1}{2} \textbf{u}_{\lambda = +1/2} +\frac{1}{2} \textbf{u}_{\lambda = - 1/2} = \textbf{u}_1 $, and we conclude that $\textbf{u}_1$ is not an extreme point. Hence, the set of extreme points of $\mathcal{C}$ is a subset of $ \left\{ \textbf{u}\in \Delta: f_1(\textbf{u}) = R_1 \right\} $. 
\end{itemize}
Since $f$ is continuous, it reaches a maximum $\textbf{u}_{max}$ on the compact $\mathcal{C}$. Then, since $f$ is convex and $\mathcal{C}$ is a convex compact, by the Krein-Milman Theorem,\footnote{ A compact convex subset of a locally convex topological vector space is the closed convex hull of the set of its extreme points. Actually, only a weaker version is used since a finite dimensional space is considered.} $\textbf{u}_{max}$ is a convex linear combination of extreme points of $\mathcal{C}$ (existence of such extreme points comes directly from the Krein-Milman theorem, since $\mathcal{C} \neq \emptyset$ ). Hence, 
 $\textbf{u}_{max} = \sum_{k=1}^n \lambda_k \textbf{u}_k$, with $\sum_{k=1}^n \lambda_k = 1$ , $\lambda_1, \lambda_2, \ldots, \lambda_n \geq 0 $ and $\textbf{u}_1, \textbf{u}_2, \ldots, \textbf{u}_n$ extreme points of $\mathcal{C}$.
 By convexity of $f$, 
\begin{align*} 
& f(\textbf{u}_{max}) \leq \displaystyle\sum_{k=1}^n \lambda_k f(\textbf{u}_k) \leq \displaystyle\sum_{k=1}^n \lambda_k f(\textbf{u}_{max}) = f(\textbf{u}_{max}),\\
% \end{equation*}
%  
%\begin{equation*}
&\text{thus }\displaystyle\sum_{k=1}^n \lambda_k (f(\textbf{u}_{max})-f(\textbf{u}_k))=0,
\end{align*}  
 which means that there exists $i\in \llbracket 1,n \rrbracket$ s.t $f(\textbf{u}_{max}) = f(\textbf{u}_i)$. We conclude that $\textbf{u}_{max}$ is an extreme point of $\mathcal{C}$. This result is known as the maximum principle~\cite{Rockafellar70}.
\subsubsection{Cardinality bound $|\mathcal{U}|\leq |\mathcal{X}|$}
This result is a special case of a more general one that we prove in Appendix \ref{C_s_range}.
\subsection{Proof of Part ii)}
The proof is partially found in \cite{Watanabe10a} and all that remains to be proved are the equality in the communication rate constraint and the range constraint $|\mathcal{U}| \leq |\mathcal{X}|$.
\subsubsection{Equality in the constraint} \label{C_s_degraded}
To prove that equality holds in the constraint for the argument of the maximum in Proposition \ref{Sufprop}.\ref{C_s_eq}, we can reuse the proof of Proposition \ref{Sufprop}.\ref{C_rec3} in Appendix~\ref{defdelta0}, so that we only need to show that $f(U) = {I}(Y;U)-{I}(Z;U)$ is convex in $p_{U|X}$. To obtain the convexity of $f$, we replace $(X,Y)$ by $(Y,Z)$ in the function $f_1$ of Lemma~\ref{fconvex}.
\subsubsection{Range constraint $|\mathcal{U}| \leq |\mathcal{X}|$} \label{C_s_range}
The proof relies on a technique used in \cite{Salehi78}. %

Define 
\begin{align*}
& \mathcal{R} \triangleq \left\{ (R,R_1):  R \geq {I}(Y;U)-{I}(Z;U),\right. \\ & \phantom{mmm}\left. R_1 \geq {I}(X;U) - {I}(Y;U),  \text{ with } U \text{---} X \text{---} Y \text{---}Z \right\}, \\
%\end{multline*} and 
%\begin{multline*}
& \mathcal{C} \triangleq \left\{ (R,R_1): R \geq {I}(Y;U)-{I}(Z;U),\right. \\ &\phantom{mmm} \left. R_1 = {I}(X;U) - {I}(Y;U),  \text{ with } U \text{---} X \text{---} Y\text{---}Z \right\}.\end{align*}
Note that the capacity region $\mathcal{C}$ is from Proposition \ref{Sufprop}.\ref{C_s_eq} and that the equality in the communication rate constraint is crucial to make it a subset of $\mathcal{R}$.
By \cite[Lemma 3]{Salehi78},
\begin{equation*}
\mathcal{R} = \left\{ (R,R_1): \forall \lambda_1, \lambda_2 \in \mathbb{R}^+, \lambda_1 R + \lambda_2 R_1 \geq G(\lambda_1, \lambda_2)\right\},
\end{equation*}
where $\forall \lambda_1, \lambda_2 \in \mathbb{R}^+$,
\begin{multline*}
G(\lambda_1, \lambda_2) \triangleq \smash{\displaystyle\inf_{ \stackrel {U \text{ s.t }} {U \text{---} X \text{---} Y \text{---} Z}}} \left[ \lambda_1 ({I}(Y;U)-{I}(Z;U)) \right. \\ \left.+ \lambda_2({I}(X;U) - {I}(Y;U)) \right].
\end{multline*}
 Consequently $G(\lambda_1, \lambda_2)$ is sufficient information to describe $\mathcal{R}$. Then, we show that for all $\lambda_1, \lambda_2 \in \mathbb{R}^+$, $G(\lambda_1, \lambda_2)$ can be achieved by considering a discrete random variable $U$ such that $|\mathcal{U}| \leq |\mathcal{X}|$.% 

Let $\lambda_1, \lambda_2 \in \mathbb{R}^+$, let $\mathcal{P}$ in \cite[Lemma 2]{Salehi78} be the $|\mathcal{X}|$-dimensional probability simplex, and let $\mathcal{X} = \left\{ x_i \right\}_{i=1}^{|\mathcal{X}|} $. Consider $\mathcal{P}$ as a set of elements of the form 
\begin{multline*}
\left(  \mathbb{P}[X=x_1 |U=u], \mathbb{P}[X=x_2 |U=u], \ldots, \right. \\ \left. \mathbb{P}[X=x_{|\mathcal{X}|} |U=u]  \right),
\end{multline*} 
with $u \in \mathcal{U}$. Then, each probability distribution on $U$ defines a measure $\mu$ on $\mathcal{P}$. Define ${H}_P(X)$, ${H}_P(Y)$, and ${H}_P(Z)$ as the entropies of $X$, $Y$, and $Z$ respectively, when the distribution of $X$ is $P\in \mathcal{P}$. Define
\begin{align*}
f_1 (P) & \triangleq \lambda_1( {H}_P(Z) - {H}_P(Y))  + \lambda_2({H}_P(Y)-{H}_P(X)) \\
f_j  (P)& \triangleq P(x_j), \text{ for } j \in \llbracket 2, |\mathcal{X}| \rrbracket.
\end{align*}
Let $P^*_X$ achieve $G(\lambda_1, \lambda_2)$, and let $\mu^*$ be such that $\int_{\mathcal{P}} P \mu^*(dP) = P^*_X$. Denote by ${H}^*(X)$ the entropy of $X$ under probability distribution $P^*_X$.
Then, by \cite[Lemma 2]{Salehi78}, there exists $P_1,P_2,\ldots ,P_{|\mathcal{X}|}$, and $\alpha_1,\alpha_2,\ldots ,\alpha_{|\mathcal{X}|}$ such that, $\sum_{i=1}^{|\mathcal{X}|} \alpha_i = 1$, 
\begin{align*}
\forall j \in \llbracket 2 , |\mathcal{X}| \rrbracket, P_X^*(x_j) 
= \int_{\mathcal{P}}  f_j(P) \mu^{*}(dP) 
= \displaystyle\sum_{i=1}^{|\mathcal{X}|} \alpha_i f_j(P_i),
\end{align*}
and,
\begin{align*}
& \lambda_1( {H}^*(Z|U)- {H}^*(Y|U)) + \lambda_2({H}^*(Y|U)-{H}^*(X|U)) \\
%&= \lambda_1 \int_{\mathcal{P}} ({H}_P(Z) - {H}_P(Y) )\mu^{*}(dP)  \\
%& \phantom{bateaummmmm}+ \lambda_2 \int_{\mathcal{P}}  ({H}_P(Y)-{H}_P(X)) \mu^{*}(dP) \\
&= \int_{\mathcal{P}} f_1(P) \mu^{*}(dP) = \displaystyle\sum_{i=1}^{|\mathcal{X}|} \alpha_i f_1(P_i).
\end{align*}
From $P_X^*(x_j) $, $ j \in \llbracket 2 , |\mathcal{X}| \rrbracket$, we can compute ${H}^*(X)$, ${H}^*(Y)$, and ${H}^*(Z)$, then
\begin{align*}
&  \lambda_1( {H}^*(Y) - {H}^*(Y|U) - {H}^*(Z) +{H}^*(Z|U)) \\
& \phantom{batman}+ \lambda_2({H}^*(X)-{H}^*(X|U) - {H}^*(Y) +{H}^*(Y|U) ) \\
& = \lambda_1( {I}^*(Y;U) -{I}^*(Z;U))+ \lambda_2({I}^*(X;U) - {I}^*(Y;U))\\
& = G(\lambda_1, \lambda_2).
\end{align*}
We have thus shown that we can choose $U$ such that $|\mathcal{U}|\leq|\mathcal{X}|$  to achieve $G(\lambda_1, \lambda_2)$. Consequently, it is enough to consider $U$ such that $|\mathcal{U}|\leq|\mathcal{X}|$, to form the set $\mathcal{R}$, as well as the set $\mathcal{C}$, since $\mathcal{C} \subset \mathcal{R}$. 
\section{Proof of Proposition \ref{discretegen}} \label{Appendixdiscretegeneral}
If $R_1 \geq {H}(X|Y)$, then by Proposition \ref{Sufprop}.\ref{C_s_eq} $C_{\textup{WSK}}(R_1,0)=\mathbb{I}(X;Y)$. Assume $R_1 \in ] 0; {H}(X|Y)[$ in the following. We note $\mathcal{X} = \{ 0, 1\}$ and by Proposition \ref{Sufprop}.\ref{C_s_eq}, we can assume $\mathcal{U} = \{ u_1, u_2\}$. We note $\beta_1 = p(X=1|U=u_1)$ and $\beta_2 = p(X=0|U=u_2)$. We can write
\begin{align}
&  {I}(U;X) - {I}(U;Y) - ( {H}(X) - {H}(Y))\nonumber \\ \nonumber
%& =   \\ \nonumber 
& =- \sum_{i=1,2} p(u_i) [{H}(X|U=u_i) - {H}(Y|U=u_i)] \\ \nonumber
& =  - \sum_{i=1,2} p(u_i) [{H}_b(\beta_i) - {H}(Y|U=u_i)] \\ 
%& = {H}(X) - {H}(Y) \\  
& = - \smash{\sum_{i=1,2}} p(u_i) \left[{H}_b(\beta_i) + {\textstyle\sum_{y \in \mathcal{Y}}} p(y|u_i) \log p(y|u_i)\right],\label{eq3set}
\end{align}
with $\forall y \in \mathcal{Y}$,
\begin{align}
& p(y|u_1) = (1-\beta_1) p(y|X=0) + \beta_1 p(y|X=1), \label{eq1set} \\
& p(y|u_2) = \beta_2 p(y|X=0) + (1-\beta_2) p(y|X=1).\label{eq2set}
\end{align}
Moreover, since the channel $p_{Y|X}$ is symmetric, there exists a permutation $\pi \in \mathfrak{S}_{|\mathcal{Y}|}$ such that 
\begin{equation}
\forall y\in \mathcal{Y}, \forall x \in \mathcal{X}, p(y|x) = p(\pi(y) | x \oplus 1), \label{permutation}
\end{equation}
where $\oplus$ denotes the modulo $2$ operation.
Thus by (\ref{eq3set}), (\ref{eq1set}), (\ref{eq2set}), (\ref{permutation}) there exists $g_{Y|X}$\footnote{The exact description of $g_{Y|X}$ is not important here, what matters is that ${H}(Y|U=u_1)$ and ${H}(Y|U=u_2)$ can be expressed with the same function.} such that ${H}(Y|U=u_1) = g_{Y|X}(\beta_1)$, ${H}(Y|U=u_2) = g_{Y|X}(\beta_2)$. Then,
\begin{multline}
 {I}(U;X) - {I}(U;Y) - ({H}(X) - {H}(Y)) \\ =  - \sum_{i=1,2} p(u_i) \left[{H}_b(\beta_i) -  g_{Y|X}(\beta_i)\right].\label{eqdev1}
\end{multline}
Similarly, by using that the channel $p_{Z|X}$ is symmetric, there exists $g_{Z|X}$ such that ${H}(Z|U=u_1) = g_{Z|X}(\beta_1)$ and ${H}(Z|U=u_2) = g_{Z|X}(\beta_2)$. Thus, we also have
\begin{multline}
 {I}(U;Y) - {I}(U;Z) - ({H}(Y) - {H}(Z)) \\ = - \sum_{i=1,2} p(u_i) \left[  g_{Y|X}(\beta_i) - g_{Z|X}(\beta_i) \right].\label{eqdev2}
\end{multline}
Consider the region $\mathcal{R}_1 \!\! \triangleq \!\!\! \displaystyle\bigcup_{\beta_0  \in  [0,1]} \!\!\! \mathcal{R}_{\beta_0}  $ and $\mathcal{R}_2 \! \triangleq \!\!\!\!\!\!\!\!\!\! \displaystyle\bigcup_{(\beta_1,\beta_2) \in [0,1]^2} \!\!\!\!\!\!\!\!\!\! \mathcal{R}_{\beta_1,\beta_2}$, with
\begin{align*}
&\mathcal{R}_{\beta_0} \!  \triangleq  \! \left\{ (R,R_1) \! : \! \right.  R \leq  {H}(Y) \! -  \! {H}(Z) \! -  \! g_{Y|X}(\beta_0)\! + \!g_{Z|X}(\beta_0) \text{, }  \\ & \phantom{lllmmmmmmm} \left. R_1 \! \leq  \! {H}(X) \! - \! {H}(Y) \! - \! {H}_b(\beta_0) \! + \!  g_{Y|X}(\beta_0)\right\}\!\!,\\
%\end{multline*}
%\begin{multline*}
&\mathcal{R}_{\beta_1,\beta_2} \triangleq \left\{ (R,R_1) :  R \leq  {I}(Y;U) -{I}(Z;U) \text{, } \right.  \\& \phantom{mmmmmmmmmmmmm} \left. R_1 \leq {I}(X;U) - {I}(Y;U)   \right\}.
\end{align*}
We can verify that both regions $\mathcal{R}_1$ and $\mathcal{R}_2$ are convex and that $\mathcal{R}_1 \subset \mathcal{R}_2$. We will use  a similar technique as in~\cite{Elgamal11}, based on Lemma \ref{convexcharact}, to show that  $\mathcal{R}_1=\mathcal{R}_2$.\footnote{Note that the fact that $R_1$ and $R$ are both lower bounds in $\mathcal{R}_1$ and $\mathcal{R}_2$ is crucial to show $\mathcal{R}_1=\mathcal{R}_2$. The same argument cannot apply if $R$ is a lower bound and $R_1$ an upper bound, whence the importance of the equality in the constraint shown in Proposition \ref{Sufprop}.\ref{C_s_eq}.} Then,
thanks to the refinement proposed in Proposition \ref{Sufprop}.\ref{C_s_eq} (equality in the constraint), we will be able to conclude for any $R_1 \in \mathbb{R}_+$,
\begin{align*}
C_{\textup{WSK}}(R_1,\!0) \!= \! \! \max_{\beta_0 \in [0,1]} \! \left( {H}(Y) \! -  \! {H}(Z) \! - \!  g_{Y|X}(\beta_0) \! + \! g_{Z|X}(\beta_0) \right)\\ \text{ such that }R_1 = H(X) \! - \! {H}(Y) \! - \! {H}_b(\beta_0) \! + \! g_{Y|X}(\beta_0).
\end{align*}

\begin{lem} [\cite{Elgamal11} \cite{Rockafellar70}] \label{convexcharact}
Let $\mathcal{C} \subset \mathbb{R}^d$ be convex. Let $\mathcal{C}_1 \subset \mathcal{C}_2$ be two bounded convex subsets of $\mathcal{C}$, closed relative to $\mathcal{C}$. If every supporting hyperplanes of $\mathcal{C}_2$ intersects with $\mathcal{C}_1$, then $\mathcal{C}_1 =\mathcal{C}_2$.
\end{lem}
Let $(R,R_1)\in \mathcal{R}_2$, and let $\alpha \in [0,1]$, then we have by~(\ref{eqdev1}),~(\ref{eqdev2})
\begin{align}
&\alpha R + (1-\alpha) R_1 \nonumber \\ \nonumber
& \leq \alpha ({I}(Y;U) -{I}(Z;U)) + (1- \alpha) ({I}(X;U) -{I}(Y;U) )\\\nonumber
& = \smash{\sum_{i =1,2}} p(u_i) [\alpha ({H}(Y) - {H}(Z) -  g_{Y|X}(\beta_i) + g_{Z|X}(\beta_i) )\\\nonumber
&  \phantom{mmlmm}   +(1-\alpha) ({H}(X) - {H}(Y) - {H}_b(\beta_i) +  g_{Y|X}(\beta_i))]  \\
& \leq \alpha ({H}(Y) - {H}(Z) -  g_{Y|X}(\beta^*) + g_{Z|X}(\beta^*) ) +(1-\alpha) \nonumber \\ 
&  \phantom{mmm} \times ({H}(X) - {H}(Y) - {H}_b(\beta^*) +  g_{Y|X}(\beta^*)), \label{betahyperplan}
\end{align}
where 
\begin{multline*}\beta^* \triangleq \arg\!\max_{\beta}  (\alpha ({H}(Y) - {H}(Z) -  g_{Y|X}(\beta) + g_{Z|X}(\beta) ) \\ +(1-\alpha) (1 - {H}(Y) - {H}_b(\beta) +  g_{Y|X}(\beta))).\end{multline*}
With the last inequality, we show that every supporting plane of $\mathcal{R}_2$ intersects $\mathcal{R}_1$. Note that the weight coefficients of $(R,R_1)$ have been taken of the form $(\alpha , 1- \alpha)$ with $\alpha \in [0,1]$, because by positivity and convexity of $\mathcal{R}_2$, we only needed to consider hyperplanes (lines) with negative slope to apply~ Lemma~\ref{convexcharact}.\\ 
Let $(R^0,R_1^0)$ be a boundary point of $\mathcal{R}_2$. There exists a supporting hyperplane $\mathcal{H}_{0}$ at $(R^0,R_1^0)$ defined by $(\alpha^0,1-\alpha^0)$.  By Equation (\ref{betahyperplan}), there exists $\beta_0^* \in [0,1]$ such that 
$$
\alpha^0 R^0 + (1-\alpha^0) R_1^0 \leq \alpha^0 R^* +(1-\alpha^0) R_1^*, 
$$
where $(R^*,R_1^*)\triangleq ({H}(Y) - {H}(Z) -  g_{Y|X}(\beta^*_0) + g_{Z|X}(\beta^*_0),{H}(X) - {H}(Y) - {H}_b(\beta^*_0) +  g_{Y|X}(\beta^*_0)) $. Then, since $(R^*,R_1^*)\in \mathcal{R}_1 \subset \mathcal{R}_2$, we also have, by definition of $\mathcal{H}_{0}$
$$
\alpha^0 R^* +(1-\alpha^0) R_1^* \leq \alpha^0 R^0 + (1-\alpha^0) R_1^0.
$$
Hence, $\alpha^0 R^* +(1-\alpha^0) R_1^* = \alpha^0 R^0 + (1-\alpha^0) R_1^0$, and thus $(R^*,R_1^*) \in \mathcal{H}_{0}$.

%
%\section{Proof of Proposition \ref{Prop_gauseq}} \label{AppendixCMS}
%\input{Appendix_cms}
%
\section{Proof of Proposition \ref{Prop_Quant}} \label{AppendixQuant}
Consider  $X  \sim  \mathcal{N}(0,\sigma_x^{2})$, $N  \sim \mathcal{N}(0,\sigma^{2}_n)$, $Y  =  X+N$. 
 We have $\sigma_y^2 = \sigma_x^{2}+ \sigma_n^{2}  $ and
\begin{multline*}
  p_{X}(x) =   \frac{1}{\sqrt{ 2 \pi \sigma_x^2}} \exp \left[ - \frac{x^2}{2 \sigma_x^2}  \right], \\ p_{X|Y}(x|y) =   \frac{1}{  \sqrt{2\pi }  } \frac{\sigma_y }{\sigma_x \sigma_n} \exp \left[ - \frac{1 }{2\sigma^{2}_n} \frac{\sigma_y^2}{\sigma_x^2}   \left(x- \frac{\sigma_x^2}{\sigma_y^2} y \right)^{2}  \right].
\end{multline*} 
Let $l\in \mathbb{N}^*$ and $k\in \llbracket 1, l\rrbracket$. Define $t_k \triangleq a( 2\tfrac{ k -1}{l-1} - 1)$ and $\Delta \triangleq \tfrac{2a}{l-1}$. Let $U$ be a scalar quantized version of $X$, defined as follows. 
\begin{align*}
 p_{U}(u_k) & \triangleq \int_{t_k}^{t_{k+1}} \!\!\!\!\!\! p_{X}(x) dx = p_X(\bar{x}_k) \Delta , \\
\forall y \in \mathcal{Y},  p_{U|Y}(u_k|y) & \triangleq  p_{X|Y}(\bar{x}_k|y)\Delta , 
\end{align*} 
where $\bar{x}_k \in [t_k , t_{k+1}]$ by the mean value theorem for integration. Hence, 
$${H}(U) =  S_U -\log \Delta, \text{ with }S_U \triangleq - \Delta \sum_k  p_{X}(\bar{x}_k) \log p_{X}(\bar{x}_k).
$$
Observe that $S_U$  is a Riemann sum that approaches $h(X) = -\int p_{X} \log p_{X}$. Thus, if we set $f(x) \triangleq -p_{X}(x) \log p_{X}(x)$, we can show that for any $a \in \mathbb{R}^+$,\footnote{We used a standard Riemann sum error bound, and erfc$(x)\leq e^{-x^2}.$}
\begin{align*}
| h(X) - S_U|  
& = \left| \int f -S_U\right| \\
& \leq \left| \int_{-\infty}^{-a} f +\int_{a}^{+\infty}  f \right|+ \left| S_U - \int_{-a}^{a} f \right|\\
& \leq \epsilon_1(a) + K_1(a) \Delta,
\end{align*}
with $K_1(a) \triangleq a \max_{[-a,a]} |f'|$, $\epsilon_1(a) \triangleq e^{-\frac{a^2}{2\sigma_x^2}}[  \alpha_1 a + \beta_1]$, and $\alpha_1, \beta_1$ constants.% $\alpha_1 \triangleq  \left( \sqrt{2\pi}\sigma_x \right)^{-1}$, $\beta_1 \triangleq \left| 1/2 - \log \alpha_1 \right|$. 

Similarly, if we define $$S_{U|Y} \triangleq - \Delta \sum_k  \int_y p_{XY}(\bar{x}_k,y) \log p_{X|Y}(\bar{x}_k|y) dy,$$ and  $g(x) \triangleq \int  p_{XY}(x,y) \log p_{X|Y}(x|y) dy,$ then, as previously, we can show that for any $a \in \mathbb{R}^+$,
\begin{align*}
| h(X|Y) - S_{U|Y}| \leq \epsilon_2(a) + K_2(a)\Delta,
\end{align*}
with $K_2(a) \triangleq a \max_{[-a,a]} |g'|$, $\epsilon_2(a) \triangleq e^{-\frac{a^2}{2\sigma_x^2}} [ \alpha_2 a + \beta_2]$, and $\alpha_2, \beta_2$ constants. % where $\alpha_2 \triangleq \tfrac{1}{\sqrt{2 \pi}} \sigma_y^2 \left(1 - (\sqrt{2}\sigma_n \gamma)^{-1} \right)^2 / \sigma_n^{3}$, $\gamma \triangleq \sigma_x / \sigma_y$,\\ $\beta_2 \triangleq  (\sqrt{\pi} /2)\alpha_2 +\left| \left(2\sqrt{2}\sigma_n\right)^{-1} \left( {\gamma^2 \sigma_n^{-2} /2}  + \log\left( {2\pi \sigma_n^{2}} \gamma^2 \right) \right) \right|$.
Thus,
\begin{multline*}
\log \Delta -(\epsilon_2(a) + K_2(a) \Delta) \\
 \leq h(X|Y) - {H}(U|Y)  \\
  \leq  \log \Delta +\epsilon_2(a) + K_2(a)\Delta.
\end{multline*}
Hence, for any $a \in \mathbb{R}^+$, if we take $\Delta$ small enough, then $|\log \Delta|  \gg \epsilon_2(a) + K_2(a) \Delta$, such that  $ h(X|Y) - {H}(U|Y) \approx \log \Delta $, and
\begin{align}
& | {I}(X;Y) - {I}(Y;U)|  \nonumber \\ \nonumber
&= | h(X) - S_U  + S_{U|Y}-h(X|Y) |\\ \nonumber
& \leq \epsilon(a) +  K(a)\Delta \\ \nonumber
& \leq \epsilon(a) +  K(a) \exp[ h(X|Y)- {H}(U|Y)] \\ \nonumber
& = \epsilon(a) +  K(a) \exp[ h(X|Y)- R_1 ] ,
\end{align}
where $\epsilon(a) \triangleq \epsilon_1(a)+\epsilon_2(a)$,  $K(a) \triangleq K_1(a)+K_2(a)$.

%If we take $a= \sigma_x \sqrt{2R_1}$ in (\ref{eqquant}), we obtain
%\begin{multline*}
%| {I}(X;Y) - {I}(Y;U)| \\ \leq  [\alpha R_1 + \beta]e^{-R_1} +  K \sqrt{R_1} e^{[ h(X|Y)- R_1 ]},
%\end{multline*}
%where $\alpha = \alpha_1 + \alpha_2$, $\beta = \beta_1 + \beta_2$, $K = (\sqrt{2}\sigma_x ) \max (| f'| + |g'|) $. We can show that $K \leq [ | \log(\sqrt{2 \pi} \sigma_x)| +11 + 4 \beta_2 +\sqrt{\pi}\alpha_2(11/\sqrt{2\sigma_n^2}-2)] / [24 \sqrt{\pi}\sigma_x^2]$.
To sum up, $\Delta$ chosen small enough ensures that ${I}(Y;U)$ approaches ${I}(X;Y)$ exponentially fast  as $R_1> h(X|Y)$ increases.%$|\log \Delta|  \gg \epsilon_2 + K_2 \Delta$, so

%-----------------------------------------%
% BIBLIOGRAPHY
%-----------------------------------------%
\bibliographystyle{IEEEtran}
\bibliography{bib}
\end{document}